\documentclass[reprint, aip, pop, floatfix, nofootinbib, eqsecnum, citeautoscript, letterpaper, onecolumn]{revtex4-2}

\usepackage{amsmath,amsfonts,amssymb}
\usepackage{amsthm,upref}
\usepackage{cancel}
\usepackage{geometry}
\usepackage{graphicx}
\usepackage{subcaption,tabularx,float}
\usepackage{hyperref}
\usepackage[capitalise]{cleveref}
\usepackage{enumerate}

\newcommand{\citeInline}[1]{Ref.~[\onlinecite{#1}]}

\newcommand{\R}{\mathbb{R}}
\newcommand{\C}{\mathbb{C}}
\newcommand{\N}{\mathbb{N}}
\newcommand{\bS}{\mathbb{S}}
\newcommand{\Z}{\mathbb{Z}}
\newcommand{\Q}{\mathbb{Q}}

\renewcommand{\H}{\mathcal{H}}
\newcommand{\eps}{\varepsilon}
\newcommand{\vphi}{\varphi}

\makeatletter
\def\ip{\@ifnextchar\bgroup{\@with}{\@without}}
\def\@with#1#2{\langle #1,#2 \rangle}
\def\@without{\langle \cdot,\cdot \rangle}
\makeatother

\newcommand{\tx}{{\tilde{x}}}
\newcommand{\ty}{{\tilde{y}}}

\renewcommand{\div}[1]{\nabla \cdot #1} % for divergence
\newcommand{\bz}{{\bar{z}}}
\newcommand{\bu}{{\bar{u}}}
\newcommand{\bZ}{{\bar{Z}}}
\newcommand{\bv}{\bar{v}}
\newcommand{\bw}{\bar{w}}

\newcommand{\hatt}{\hat{t}}
\newcommand{\hn}{\hat{n}}
\newcommand{\hb}{\hat{b}}
\newcommand{\ta}{\tilde{\alpha}}
\newcommand{\iotab}{{\iota\hspace{-0.6em}\raisebox{0.75pt}{\,--}}}
\newcommand{\Tee}{T}

\newcommand{\T}{\mathbb{T}}

\DeclareMathOperator{\Diff}{Diff}

\DeclareMathOperator{\Ima}{Im}
\DeclareMathOperator{\Ker}{Ker}
\DeclareMathOperator{\re}{re}
\DeclareMathOperator{\im}{im}

\DeclareMathOperator{\sech}{sech}
\DeclareMathOperator{\csch}{csch}

\theoremstyle{plain}
\newtheorem{thm}{Theorem}[section]
\newtheorem{proposition}[thm]{Proposition}
\newtheorem{corollary}[thm]{Corollary}
\newtheorem{lemma}[thm]{Lemma}

\theoremstyle{definition}
\newtheorem{definition}{Definition}

\begin{document}
\title{Normal Forms and Near-Axis Expansions for Beltrami Magnetic Fields}
%\author{Nathan Duignan \& James~D. Meiss}
%\address{Department of Applied Mathematics, University of Colorado, Boulder, CO, 80309-0526 USA}

%REVTEX VERSION
\author{Nathan Duignan and James~D. Meiss}
\affiliation{Department of Applied Mathematics, University of Colorado, Boulder, CO, 80309-0526 USA}
\thanks{Email correspondance: nathan.duignan@colorado.edu}

\begin{abstract}
A formal series transformation to Birkhoff-Gustavson normal form
is obtained for toroidal magnetic field configurations in the neighborhood of a magnetic axis. Bishop's rotation-minimizing coordinates are used to obtain a local orthogonal frame near the axis in which the metric is diagonal, even if the curvature has zeros. 
We treat the cases of vacuum and force-free (Beltrami) fields in a unified way, noting that the vector potential is essentially the Poincar\'e-Liouville one-form of Hamiltonian dynamics, and the resulting magnetic field corresponds to the canonical two-form of a nonautonomous one-degree-of-freedom system. Canonical coordinates are obtained and Floquet theory is used to transform to a frame in which the lowest-order Hamiltonian is autonomous. The
resulting magnetic axis can be elliptic or hyperbolic, and resonant elliptic cases are treated.
The resulting expansion for the field is shown to be well-defined to all orders, and is explicitly computed to degree four. An example is given for an axis with constant torsion near a $1:3$ resonance.
\end{abstract}
\date{\today}

\maketitle

%\tableofcontents %remove for REVTEX

\section{Introduction}
The utility of a device for confining plasma by a magnetic field depends crucially on the geometry of the field,
especially for the case of toroidal confinement, like that in tokamaks or stellarators.
Any such configuration should ensure that plasma pressure and electromagnetic forces balance to obtain an MHD equilibrium. Additional considerations such as omnigenity or quasisymmetry are necessary to ensure good confinement of gyrating particles (see the review by \citeInline{Helander14}). For any of these desired properties, two crucial questions must be answered. Firstly, do magnetic fields with the desired property exist? If so, what is the topology, geometry, and dynamics of these fields? Many theoretical tools have been constructed to gain insight into these questions. One classical tool, which is enjoying a recent resurgence, is the near-axis expansion.

%% What is a near-axis expansion?
In essence, a near-axis expansion is a method for computing the magnetic field $B$ as a power series in distance from a \emph{magnetic axis}; a closed field line $r_0: \bS^1 \to \R^3$ of $B$ that is an isolated fixed point of its Poincar\'{e} first-return map. In the axisymmetric case, such an axis is a circle at the center of a nested family of toroidal magnetic surfaces, but more generally there may not be a smooth family of such surfaces. Two techniques have been developed for near-axis expansions: the direct and the inverse method. 

%% Direct and Inverse briefly
The direct method was pioneered by Mercier \cite{Mercier64} and Solov'ev and Shafronov \cite{solovevPlasmaConfinementClosed1970} for studying solutions to the force-balance equations
\begin{equation}\label{eq:forceBalance}
	J \times B = \nabla p ,\qquad \nabla\cdot B = 0,
\end{equation}
where $J = \nabla\times B$ is the current vector and $p$ is the (scalar) plasma pressure. The core idea is to use a Frenet-Serret frame based on $r_0$ to obtain what are now called Mercier coordinates $(\rho,\theta,s)\in\R^+\times \T^2$. In these coordinates the axis $r_0(s)$ is simply $\rho = 0$. All physical quantities are then expanded as formal power series in $\rho$ and \eqref{eq:forceBalance} is solved order by order.
Key goals of the direct method are to establish formal solutions to \eqref{eq:forceBalance} (or, perhaps equally as interestingly, imply their non-existence), and to compute an integral of the system $\psi$, e.g., the toroidal magnetic flux, in terms of the Mercier coordinates. The direct method is beneficial when no assumptions are made about the possible topology, geometry, or dynamics of $B$. Since the pioneering work, the direct method has been implemented by many authors \cite{lortzEquilibriumStabilityThreeDimensional1976,lortzEquilibriumStabilityEll1977,Bermardin85,bernardinIsodynamicalOmnigenousEquilibrium1986,salatNonexistenceMagnetohydrodynamicEquilibria1995,Weitzner16,Landreman18b,Landreman19d,Jorge20a,jorgeUseNearaxisMagnetic2021}.

%% Inverse
In contrast, the inverse method, as used most prominently by Garren and Boozer \cite{Garren91a,Garren91b}, but appearing earlier in \citeInline{lortzEquilibriumStabilityThreeDimensional1976}, assumes the existence of special magnetic coordinates such as Boozer or Hamada coordinates \cite{Helander14}. These consist of an integral $\psi$ and a pair of angles $\theta,\phi$, so that the contravariant components of $B$, for example, depend only on $\psi$. The core aim of the inverse method is to determine the Euclidean coordinates $(x,y,z)$ as a series expansion in the magnetic coordinates. The benefit of this method is that it can efficiently provide expressions for physical quantities in terms of magnetic coordinates, and these, in turn are useful for further theoretical exploration. However, since the inverse method necessarily assumes the existence of magnetic coordinates, it implicitly assumes that the field line flow is integrable (see, for instance, \citeInline{burby2021integrability}). Conversely, if a nonvanishing magnetic field has toroidal flux surfaces, then, in the neighborhood of any flux surface, there exist magnetic coordinates \cite{kruskalEquilibriumMagneticallyConfined1958}; this construction can be extended to a neighborhood of an axis as well \cite{burby2021integrability}. Nevertheless, given a magnetic axis $r_0$, there may not exist a local, integrable field $B$. Indeed, an outstanding conjecture of Grad is that smooth solutions to \eqref{eq:forceBalance} do not exist for a general toroidal domain \cite{gradToroidalContainmentPlasma1967}.

% Brief Summary of our work
In this paper we study the direct method for near-axis expansion using a Hamiltonian perspective. Any nonzero, divergence-free vector field can be locally written as a  non-autonomous $1\tfrac12$ degree Hamiltonian system \cite{Hazeltine03}. The true power of this Hamiltonian reformulation is that all information about the vector field $B$ is stored in a single function, the Hamiltonian $H$. Consequently, the dynamics of the field lines of $B$ can be understood through this single function. Moreover, the perspective lends itself to the many ideas and tools of Hamiltonian mechanics and more generally, of symplectic geometry. In this paper, we demonstrate the benefits of this view through novel applications of classical ideas of Hamiltonian mechanics to near-axis expansions.

% We treat Beltrami
A similar perspective was adopted by Bernardin, Moses, and Tataronis \cite{Bermardin85,bernardinIsodynamicalOmnigenousEquilibrium1986} to investigate magnetic fields satisfying \eqref{eq:forceBalance} assuming that $\nabla p \neq 0$.
In contrast to these papers, we treat \eqref{eq:forceBalance} under the assumption that $\nabla p = 0$ in a neighborhood of the axis. In this case, the current must be parallel to the field,
\begin{equation}
	\nabla\times B = k B.
\end{equation}  
When $k\neq 0 $ such a field is called \emph{Beltrami} (or \emph{force-free}); the \emph{vacuum} case corresponds to 
$k = 0$. In these cases, the field lines are generically chaotic, as was first emphasized by Arnold in the fluid context \cite{Arnold65b} 
(see also \citeInline{etnyreContactTopologyHydrodynamics2000a,encisoBeltramiFieldsExhibit2020,cardonaConstructingTuringComplete2021}).
Since flux surfaces do not generically exist, the inverse method cannot be used. The vacuum field case has been treated previously in \cite{Jorge20a}. There, the authors assume the existence of a magnetic potential $\phi$ such that $B = \nabla \phi$, which of course implies that $J = 0$. Instead, we consider the vector potential $A$ such that $B = \nabla\times A$, allowing, for the first time, a unified expansion for both Beltrami and vacuum fields. The expansion is given explicitly, to all orders, in \cref{prop:FormalSeriesSolution}.

% We use differential forms; revealing the intrinsic geometry
Our work further differs from \citeInline{bernardinIsodynamicalOmnigenousEquilibrium1986} by recasting the expansion in terms of differential forms. A tutorial on differential forms specifically for plasma physics is given in \citeInline{MacKay20}. By translating the theory into the language of differential forms we reveal the intrinsic geometry of vacuum and Beltrami fields: they give $M$ the structure of a cosymplectic and contact manifold, respectively. For Beltrami fields, the utility of a contact structure was first understood by Etnyre and Ghrist \cite{etnyreContactTopologyHydrodynamics2000a} and since has been the source of many interesting results, most recently \citeInline{cardonaConstructingTuringComplete2021}. As far as we are aware, the result that vacuum fields are cosymplectic is novel. While the work here does not crucially depend on the understanding of these geometries, we believe that there can be further synergies between symplectic topology and plasma physics.

% Floquet theory; both elliptic and hyperbolic axis
We will apply two useful tools from Hamiltonian mechanics: Floquet and normal form theory.  Floquet theory \cite{floquetEquationsDifferentiellesLineaires1883} is the study of time-periodic, linear differential equations
and was specialized to the Hamiltonian case by Moser \cite{Moser58b}. It provides a canonical coordinate system in which the linear system becomes autonomous, thus giving an efficient way to compute its stability.  In our context, the leading order terms in the near-axis expansion become independent of the toroidal angle, and the axis is revealed to be hyperbolic or elliptic. The Floquet transformation was implicitly computed in \citeInline{Mercier64,solovevPlasmaConfinementClosed1970,Bermardin85,Jorge20a} as a sequence of transformations based on geometric assumptions about the flux surfaces near the axis. 
As we will demonstrate, the composition of these transformations is indeed the Floquet transformation. 
An important result of Floquet theory is that when the axis is elliptic, its rotational transform, $\iotab_0$, is related to the torsion $\tau$ of the curve $r_0$; we will show this holds for the Beltrami case as well. Moreover, our results also hold for hyperbolic axes which have stable and unstable manifolds with ``expansivity'' $\nu_0$. Such configurations are of importance in the design of divertors \cite{Boozer15,Boozer18}. 

% Normal Form Theory
Normal form theory for Hamiltonian systems was pioneered by Birkhoff \cite{Birkhoff27}. The theory gives a way to compute ``simple'' coordinates in the neighborhood of a periodic orbit; a nice exposition is given in \citeInline{meyerIntroductionHamiltonianDynamical2009}. We will apply this technique to near-axis expansions. Essentially, normal form theory provides an iterative procedure to remove as many terms in a power series expansion of the Hamiltonian as possible. If the axis is elliptic and non-resonant, that is if $\iotab_0 \notin \Q$, or if the axis is hyperbolic, then normal form theory gives coordinates $(x,y,s)\in D^2 \times \bS^1$ so that $H$ is (formally) of the form $H(x^2+y^2)$ or $H(x y)$, respectively. If the axis is resonant with $\iotab_0 = p/q$ then Gustavson's normal form theory gives coordinates $(\rho,\theta,\phi)\in \R^+\times \T^2$ so that $H = \tfrac12\tfrac{p}{q} \rho^2 + K(\rho,q\theta + p\phi)$ \cite{Gustavson66}. In each case, $H$ is formally integrable: normal form theory provides both simple coordinates and an approximate integral. If the normal form series converges, these coordinates give a true integral, defining flux surfaces, even in the resonant case.

Our normal form results should be directly compared to previous work for the nonresonant elliptic case \cite{Bermardin85}; these authors compute an adiabatic invariant near the axis. Their method uses generating functions to implicitly give the coordinate transformation. 
A similar procedure was used in \citeInline{Jorge20a} to compute flux surfaces for a vacuum field; their flux coordinate $\psi$ is, in essence, the adiabatic invariant of Bernardin et al. As we will show, normal form theory applies to this case, but also applies to hyperbolic and resonant elliptic axes. Moreover, we will use a near-resonant normal form \cite{meyerIntroductionHamiltonianDynamical2009} to give approximate flux surfaces when the on-axis rotational transform $\iotab_0$ is near a low order rational. A key difference from the generating function method is our use of Lie series to compute the normalizing transformation, in line with \citeInline{dragtLieSeriesInvariant1976}. As they argued, the crucial benefit Lie series provide over the generating functions is efficiency as well as the ease of computing the inverse. 

% Minor refinements like Bishop frame
The paper is outlined as follows. In \cref{sec:background}, Beltrami and vacuum fields are 
introduced through the lens of differential forms. This translation from vector 
calculus notation establishes the intrinsic geometry of vacuum and Beltrami fields.
In \cref{sec:FramingAxis} the magnetic axis is defined
and Bishop's coordinates \cite{Bishop75} are introduced.
These give Mercier coordinates without the assumption of non-vanishing curvature.
A further advantage of these coordinates is that the metric is diagonal.
In \cref{sec:FloquetAndNormalForm} the Hamiltonian formulation is given and 
the classical theory of Floquet and of normal forms, including the near-resonant case, is recalled. 
\cref{sec:ApplicationsToMagneticFields} contains the application to near-axis expansions
and the formal expansion for the Hamiltonian for Beltrami fields to all orders is found in \cref{prop:FormalSeriesSolution}.
Lastly, we apply the Floquet transformation and deconstruct it into the geometric transformations of previous work.
Finally, in \cref{sec:Examples} we give two examples of the normal form computation. Our examples use discrete symmetry to obtain closed curves. We apply this method to obtain a family of curves with constant torsion.
These examples are chosen so that the axis is elliptic and the on-axis rotational transform $\iotab_0$ is arbitrarily close to a $1:3$ resonance. The first example uses a regular normal form, while the second uses the near-resonant normal form. The calculated approximate integrals are then compared to the true field line dynamics. Future directions and concluding remarks are given in \cref{sec:conclusion}.
	
%%%%%%%%%%
%%%%% Background
%%%%%%%%%%	
\section{Geometry of Vacuum and Beltrami fields}\label{sec:background}

%%%%%%
%%% Differential Forms
%%%%%%

In this paper we will consider a solid torus $D^2 \times \bS^1$ in $\R^3$, with the Euclidean metric and the standard volume form. However, the equations defining a vacuum or Beltrami field can be given for any three-manifold $M$, with arbitrary Riemannian metric $g$ and corresponding volume form $\Omega$. In this section we give this general description through the use of differential forms, which reveals their intrinsic geometry. A summary of the translation is given in \cref{tab:translations} and further exposition is given by MacKay \cite{MacKay20}.

Suppose that $M$ is an orientable three-dimensional manifold with metric $g$ and Riemannian volume form $\Omega$. Associated with any non-vanishing magnetic field $B$ on $M$ is the so-called \textit{flux form}; a two-form $\beta$ defined by taking the interior product of $B$ with the volume form
\begin{equation}\label{eq:beta}
	\beta := \iota_B \Omega.
\end{equation} 
The name follows from the fact that, given any two-dimensional surface $S$ in $M$, the magnetic flux through $S$ is given by $\int_S \beta$. 

The requirement that magnetic fields are divergence free, $\nabla \cdot B = 0$, can be restated in terms of the flux form $\beta$ as $d\beta = 0$, that is, that $\beta$ is closed. If $B$ is non-vanishing, it also follows that $\beta$ has maximal rank. Any two-form that is both closed and of maximal rank is called \emph{presymplectic} \cite{burby2021integrability}.
Conversely, as shown in \citeInline{burby2021integrability}, given any presymplectic form $\beta$, there exists a unique, non-vanishing, divergence-free vector field $B$ such that $\iota_B \Omega =\beta$.
Hence, the magnetic field $B$ and flux-form $\beta$ are dual views of the same mathematical object. 

With the metric $g$ in hand, there is a third view of a magnetic field. This is provided through the musical isomorphisms relating one-forms to vector fields, namely,
\[ B^\flat := \iota_B g = g(B,\cdot). \]
One can think of $B^\flat$ as the covariant representation of $B$. A useful relationship between $B^\flat$ and $\beta$ is given by the Hodge star operator.
In an arbitrary coordinate system $(x^1,x^2,x^3) \in M$, $B^\flat$ is the covariant representation of a magnetic field as a one-form:
\begin{equation}\label{eq:Bflat}
	B^{\flat} = B_i dx^i .
\end{equation}

The relationship between $B^\flat$ and $\beta$ is given through the Hodge star operator $\star$, which
provides an isomorphism between $k$-forms and $(3-k)$-forms. In local coordinates, the operator
is defined on two-forms $\beta$ as
\begin{equation}\label{eq:HodgeStarTwoForm}
	\beta = \tfrac12 \epsilon_{ijk} \beta^i dx^j \wedge dx^k \,\longmapsto\,
	\star\beta =  \frac{1}{\rho} g_{ij} \beta^j dx^i ,
\end{equation}
and on one-forms as
\begin{equation}\label{eq:HodgeStarOneForm}
	\alpha = \alpha_i dx^i \,\longmapsto\,
	\star \alpha = \tfrac12\epsilon_{ijk} \rho g^{il} \alpha_l dx^j \wedge dx^k ,
\end{equation}
where $\rho = \sqrt{\det g_{ij}}$. 
The correspondence between $B^\flat$ and $\beta$ is then
\begin{equation}\label{eq:HodgeFlat}
	\beta = \star B^\flat.
\end{equation}

It is well known for $M =\R^3$ that any divergence-free vector field has a vector potential $ A$: $B = \nabla\times A$. This result for differential forms becomes: since $\beta$ is closed, and all closed two-forms on $\R^3$ are exact, there is a primitive one-form $A^\flat = \alpha$ for $\beta$:
\begin{equation}\label{eq:BetaExact}
	 \beta = d\alpha.
\end{equation}
Using \eqref{eq:HodgeFlat} this is also written
\begin{equation}\label{eq:alpha}
	B^\flat = \star d \alpha.
\end{equation}
More generally, the vector potential exists for any manifold $M$ on which closed two forms are exact \cite{arnoldMathematicalAspectsClassical2006}.

Given some additional structure on $B$, e.g., if it obeys magneto-hydrostatics (MHS), is Beltrami, or is a vacuum field, then there is a corresponding geometric interpretation. To see this, firstly note that the current $J$,  defined by $J = \nabla \times B $, becomes $ \iota_J \Omega = dB^\flat$. If $B$ satisfies MHS, then there must exist $p$ such that $J\times B = \nabla p$, or equivalently
\[
	\iota_J \beta = -dp.
\]
In open regions where $dp \neq 0$, $(B,J,p)$ is an \textit{integrable presymplectic system}, see \citeInline{burby2021integrability} for details.

Alternatively, if $B$ is a vacuum field then $J = 0$, so $dB^\flat = 0$. Thus the one-form $B^\flat$ is closed and  $\beta \wedge B^\flat$ is a volume form on $M$. A manifold $M$ together with a presymplectic form $\beta$ and a closed one-form $\eta$ such that $\beta \wedge \eta $ is a volume form, is called a \textit{cosymplectic manifold}. It follows that $B$ is a vacuum magnetic field if $(M,\beta, B^\flat)$ is a \emph{cosymplectic structure} on $M$. This places vacuum fields in the realm of cosymplectic geometry.

Lastly, for a Beltrami (force-free) field $J = kB$. Translating to differential forms gives 
\begin{equation}\label{eq:Beltrami}
	dB^\flat = \iota_J \Omega = k \iota_B\Omega = k \beta.
\end{equation}
In this paper we will assume that $k$ is constant. A manifold $M$ together with an exact, presymplectic form 
$\beta = d\eta$, so that $\beta \wedge \eta$ is a volume form, is called a \textit{contact manifold}. 
Since \eqref{eq:Beltrami} implies that $\eta = k^{-1}B^\flat$, is indeed a primitive, 
$d\eta =   d(k^{-1}B^\flat) = \beta$, it follows that when $B$ is a Beltrami  
$(M,\beta, k^{-1} B^\flat)$ is a contact manifold. 
This places Beltrami fields in the realm of contact geometry.
	
We summarize these ideas as:
\begin{lemma}\label{lem:FieldConditions}
Suppose $B$ is a non-vanishing magnetic field on an orientable three-manifold $M$ with volume form $\Omega$. Let $\beta = \iota_B \Omega$. Then:
\begin{enumerate}
	\item $B$ is MHS if $ (M,\beta) $ is a presymplectic manifold, 
		$\iota_J \Omega = dB^\flat$ and $\iota_J\beta = -dp$;
	\item $B$ is a vacuum field if $(M,\beta,B^\flat)$ is a cosymplectic manifold; and
	\item $B$ is a Beltrami field if there exists $k\neq0$ such that $(M,\beta, k^{-1}B^\flat)$ 
		is a contact manifold.
\end{enumerate}
\end{lemma}
	
By viewing magnetic dynamics  in this way, one can not only instantly see the differing geometries of these three cases, but also the relationship between magnetic fields and Hamiltonian mechanics. 
This relationship will be used heavily below. 
This geometric view of magnetic fields is not new and many interesting properties of magnetic fields have already been uncovered through this perspective. See, for instance, \citeInline{etnyreContactTopologyHydrodynamics2000a,encisoBeltramiFieldsExhibit2020,cardonaConstructingTuringComplete2021}.
	
The Beltrami condition \eqref{eq:Beltrami} can be reformulated as a \textsc{pde} for 
the coefficients of the vector potential $\alpha$. Indeed, this will be used to compute the Hamiltonian.
Using \eqref{eq:BetaExact} with \eqref{eq:Beltrami} requires that
\[
	k d\alpha = d B^\flat,\qquad \iota_B \Omega = d\alpha .
\]
Then from \eqref{eq:HodgeFlat}, $B^\flat = \star d\alpha$, which implies
\begin{equation}\label{eq:BeltramiAlpha}
	 kB^\flat = \star d (k\alpha) = (\star d)^2 \alpha.
\end{equation}
Equivalently,  $(\star d - k)\alpha = \vartheta$ where $\vartheta$ is some closed one-form. However, note that it is not $\alpha$ but $d\alpha$ that defines the field-line dynamics: gauge freedom implies any closed one-form can be added to $\alpha$ without changing $B$. Thus without loss of generality, we could set $\vartheta=0$. Nevertheless, we will retain \eqref{eq:BeltramiAlpha} because, as will be seen, it is more useful to use the gauge freedom to select a desired form for $\alpha$.
	
From \cref{lem:FieldConditions}, the vacuum field case implies
\begin{equation}\label{eq:VacuumAlpha}
	d B^\flat = 0 \quad \implies \quad (\star d)^2 \alpha = 0.
\end{equation}
Of course, this is exactly the Beltrami equation \eqref{eq:BeltramiAlpha} with $k = 0$. This enables the simultaneous treatment of vacuum and Beltrami fields; simply treat the Beltrami case and then let $k \to 0$.		
	
Thus, the fundamental system to solve is \eqref{eq:BeltramiAlpha}. We will expand this \textsc{pde} in the neighborhood of a magnetic axis, order-by-order in the radius to obtain an explicit construction of a normal form and a relation to Hamiltonian dynamics in \cref{sec:ApplicationsToMagneticFields}. 

{\renewcommand{\arraystretch}{1.3} %giving the table a little vertical padding
\begin{table}[htp]
\begin{center}
\begin{tabular}{l|c|c}
 				& Vector Calculus & Differential Forms \\
\hline
Metric			& $g^{ij} = \nabla x^i \cdot \nabla x^j$ 			& $ds^2 = g_{ij} dx^i dx^j$ \\
Volume			& $\rho = \sqrt{\det g_{ij}}$						& $\Omega = \rho dx^1 \wedge dx^2 \wedge dx^3$ = $\star 1$ \\
Covariant		& $\vec{B} = B_j \nabla x^j$ 						& $B^\flat = B_j dx^j$ \\
Contravariant 	& $B^i = \vec{B} \cdot \nabla x^i$  				& $B = B^i \partial_{x^i}$ \\
%				& $\vec{B} = \epsilon_{ijk}\beta^k \nabla x^i\times \nabla x^j$	
%										& $ \beta = \epsilon_{ijk}\beta^k dx^i\wedge dx^j $ \\
				& $\vec{B} = \tfrac12 \epsilon_{ijk} \rho B^i \nabla x^j\times \nabla x^k$	
										& $ \beta = \tfrac12 \epsilon_{ijk}\rho B^i dx^j\wedge dx^k $ \\			
Hodge Star		& $B_i = g_{ij} B^j$								& $B^\flat = \star \beta$  \\
%Curl 			& $\nabla \times \vec{B}$   						& $dB^\flat$ \\
Divergence 		& $\nabla \cdot \vec{B}$   							& $d \star B^\flat = d\iota_B \Omega$ \\
\hline
Flux			& $\vec{B} \cdot d^2 S$								& $\beta = \iota_B \Omega = \star B^\flat$ \\
Current			& $\vec{J} = \nabla \times \vec{B}$					& $J^\flat =  \star d B^\flat$ \\
Vector Potential& $\vec{B} = \nabla \times \vec{A}$					& $B^\flat = \star d \alpha$ ($A^\flat \equiv \alpha$) \\
\end{tabular}
\end{center}
\caption{Translations between vector calculus and differential forms.
We use the summation convention.
Note that that matrix $(g_{ij})$ is the inverse of $(g^{ij})$}
\label{tab:translations} 
\end{table}%
}

%%%%%%%%%%%
%%%% The Axis and Coordinates
%%%%%%%%%%%
\section{Coordinates near a Magnetic axes}\label{sec:FramingAxis}

%%%%%%%%%%%
%%%% The Axis
%%%%%%%%%%%
\subsection{Magnetic axes}
Magnetic axes are unavoidable in the study of plasma confinement since most containment
designs are based on toroidal geometry. Such a device must have an axis that is a closed field line.
In the simplest case this is the ``center'' of family of nested toroidal surfaces. However, any definition must not assume integrability and exclude closed field lines on rational tori.

Generally, suppose that
$r_0: \bS^1 \to \R^3$ is a closed field line of a nonzero, smooth magnetic field $B$.
Let $U = D^2 \times \bS^1$ be a tubular neighborhood of the axis and $\Sigma$ be some local section transverse to $B$ containing a point $z\in r_0$. The flow of $B$ produces a well-defined Poincar\'{e} first-return map $\pi_{\Sigma}:\Sigma\to\Sigma$ with a fixed point $z$. The local dynamics of the closed field line $r_0$ can be characterized by the dynamics of the map $\pi_\Sigma$. 
	
Using the Poincar\'e map we can exclude closed orbits on rational surfaces from our definition of a magnetic axis as follows.
%%%%%
\begin{definition}\label{def:MagneticAxis}
A closed field line $r_0:\bS^1 \to \R^3$ is a \emph{magnetic axis} if each point $z \in r_0$ is an isolated fixed point of its Poincar\'{e} first-return map, $\pi_\Sigma$. 
%	If any of the eigenvalues of $D\pi_\Sigma$ are equal to $1$ at the fixed point, then $r_0$ is said to be \emph{degenerate}.  
\end{definition}
%%%%%

This condition is coordinate independent and does not depend on the choice of section $\Sigma$. 
Indeed the flow of $B$ provides a conjugacy between the first-return maps on any pair of sections \cite{Meiss17a}.
%Indeed, given any other section, say $\Sigma^\prime$ containing $z^\prime \in r_0$, the flow of the smooth vector field $B$ generates a diffeomorphism $\Sigma \to \Sigma^\prime$ by flowing points on $\Sigma$ until they hit $\Sigma^\prime$. This diffeomorphism provides a conjugacy between the first return map associated with $\Sigma$ and that associated with $\Sigma^\prime$. Hence, when $z$ is an isolated fixed point of $\pi_\Sigma$ so is $z^\prime$  of $\pi_\Sigma^\prime$; moreover, the eigenvalues of the respective Jacobians coincide. 
		
However, as is sketched in \cref{fig:magAxisExam}, there could be several such axes, perhaps of differing local topology. In \cref{sec:normalFormsSetUp}, the notion of a \emph{degenerate} and \emph{nondegenerate} magnetic axes is defined. As will be seen, a nondegenerate axis must be elliptic or hyperbolic.
	
%%%%%%%
\begin{figure}[ht]
	\centering
	\includegraphics[width=0.6\linewidth]{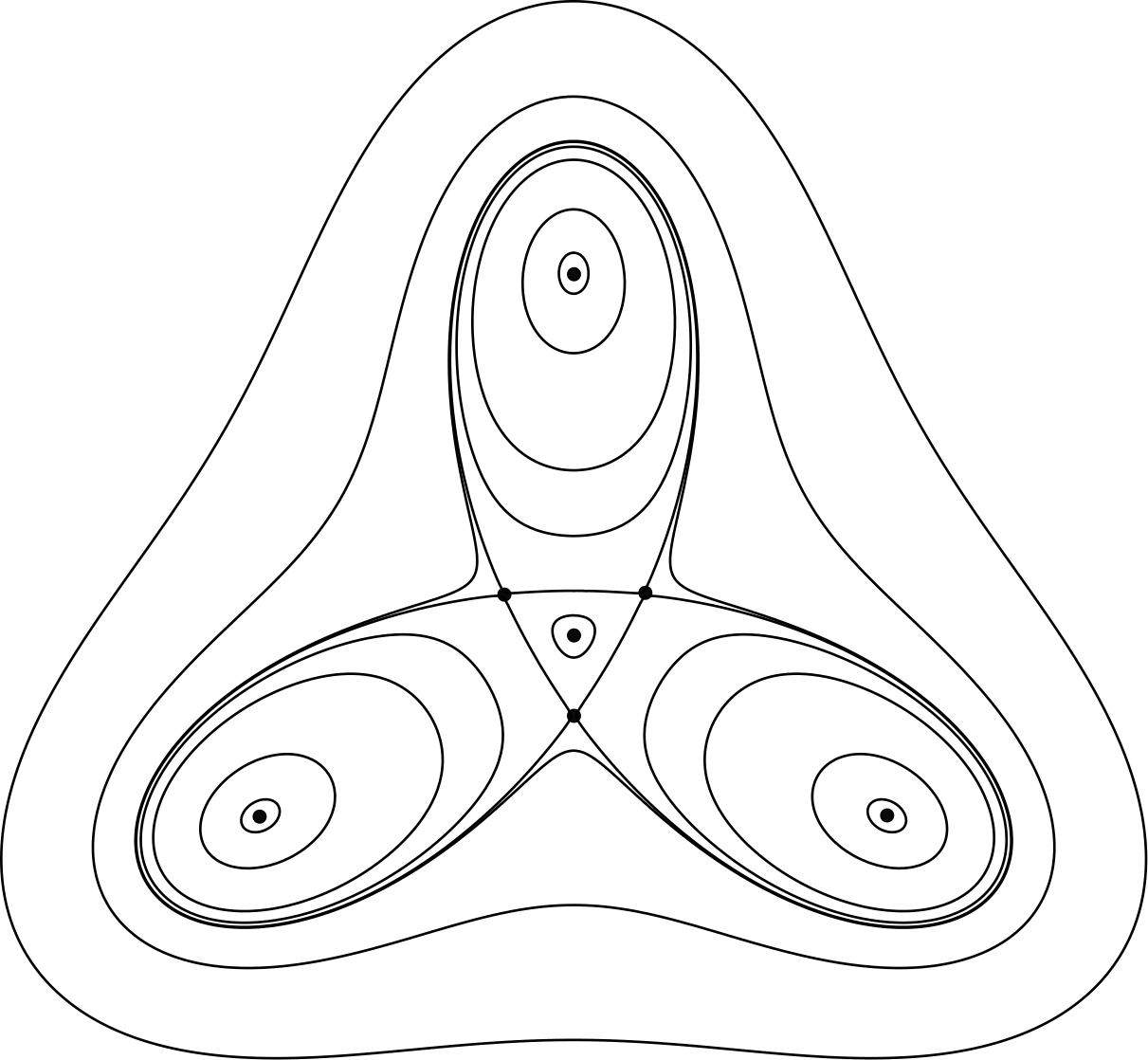}
	\caption{A global Poincar\'{e} section of an example magnetic field. There are seven ``magnetic axes'' 
		by  \cref{def:MagneticAxis}, four are elliptic and three, hyperbolic. 
		Perhaps only one might be called  the``central axis.''  }
	\label{fig:magAxisExam}
\end{figure}
%%%%%%%

%%%%%%%%%%%
%%%%% Frames
%%%%%%%%%%
\subsection{Framing a magnetic axis}

In order to understand the possible field line behavior in the neighborhood of a magnetic axis $ r_0(s) $, it is useful to have good coordinates defined in its neighborhood. As first demonstrated by Mercier \cite{Mercier64}, when $r_0 \in C^3([0,\Tee), \R^3)$ and the curvature of $r_0$ is non-vanishing, these can be provided through the Frenet-Serret \emph{moving frame} (see, for instance, \citeInline{Bishop75}). Specifically, when $s$ is the arc length, define the unit tangent, $\hatt(s) = r_0'$, normal $\hn(s)$, and binormal $\hb(s)$ vectors. Taking these to be row vectors, they satisfy the matrix \textsc{ode}
\begin{equation}\label{eq:FrenetSerret}
	\frac{d}{ds} \left(\begin{array}{c}
		\hatt \\ \hn \\ \hb
		\end{array}\right)
		=  \left(\begin{array}{ccc}
	        	0 & \kappa(s) & 0 \\
	        	-\kappa(s) & 0 & \tau(s) \\
	        	0 & -\tau(s) & 0
	        	\end{array}\right)\left(\begin{array}{c}
	        	\hatt \\ \hn \\ \hb 
	    \end{array}\right).
\end{equation}
Here $\kappa(s)$ and $\tau(s)$ are the curvature and torsion of $r_0$, respectively. 
Under the assumption that $\kappa(s)$ is non-vanishing, they are given explicitly by
\[	
	\kappa = |r_0^{\prime\prime}|,\qquad \tau 
	  = \frac{(r_0^\prime \times r_0^{\prime\prime})\cdot r_0^{\prime\prime\prime}}{\kappa^2}.
\]
		
The Frenet-Serret frame defines a local embedding
\begin{equation}\label{eq:EuclidMapping}
	\pi_{FS} : D^2\times \bS^1 \to \R^3,\qquad (x,y,s) \mapsto r_0(s)+ x \hn(s) + y \hb(s).
\end{equation}
In other words, $\pi_{FS}$ is an embedding of the trivial disk bundle $ D^2\times \bS^1$
into a tubular neighborhood of $r_0(s)$ in $\R^3$. In the plasma physics literature, these coordinates are often referred to as Mercier coordinates \cite{Jorge20a}.

While the Frenet-Serret frame constructs coordinates in terms of the geometrically significant
quantities $\kappa$ and $\tau$, in some practical cases this frame does not exist.
This occurs, for example, if $r_0$ is not $C^3$, or, more crucially, if $r_0$ has any
inflection points or straight segments, i.e., points with $\kappa = 0$. 
		
There are other choices for an orthonormal 
frame based on the curve $r_0(s)$. Such a frame can also be obtained that has a 
diagonal induced metric (in contrast to the Frenet-Serret case). Such a frame with is called \emph{rotation minimizing} \cite{Bishop75}. 
	 
Note that Mercier \cite{Mercier64} established a rotation minimizing frame starting with the Frenet-Serret frame. However, the former can be constructed independently of the existence of the latter. 
There are at least two ways to do this. The first is to use a 
three-dimensional version of Fermi-Walker transport \cite{dandoloffParallelTransportSpace1989}.
%It is a general technique with applications to higher dimensions or non-Euclidean metrics.
The second is Bishop's \emph{relatively parallel adapted frame} \cite{Bishop75}. 
A \emph{relatively parallel} vector field, $v(s)$, is one that is \textit{normal} to the curve, that is $v(s) \cdot \hatt(s) = 0$, but such that $v^\prime(s)$ is \textit{parallel} to $\hatt(s)$. Provided that $r_0$ is at least $C^2$, there exists a unique relatively parallel vector field $v(s)$ such that $v(0) = v_0$ for every initial normal vector $v_0$, see Thm. 1 of \citeInline{Bishop75}. These vector fields can be constructed from a Frenet-Serret frame; however, they may also be constructed from any orthonormal frame. Crucially, this means that the curvature need not be nonzero, and the curve need not be $C^3$.
	
As a consequence, for each initially orthonormal basis $(\hatt(0), \hn_1(0), \hn_2(0))$ we can compute a unique, relatively parallel adapted orthonormal frame $(\hatt,\hn_1,\hn_2)$ along the curve $r_0$. Being relatively parallel, $\hn_1^\prime = -\kappa_1 \hatt,\, \hn_2^\prime = -\kappa_2 \hatt$ for some functions $\kappa_1,\kappa_2$. Thus
\begin{equation}\label{eq:BishopFrame}
	\left(\begin{array}{c}
		\hatt^\prime \\ \hn_1^\prime \\ \hn_2^\prime
		\end{array}\right)
	= \left(\begin{array}{ccc}
			0 & \kappa_1(s) & \kappa_2(s) \\
			-\kappa_1(s) & 0 & 0 \\
			-\kappa_2(s) & 0 & 0
	\end{array}\right)\left(\begin{array}{c}
							\hatt \\ \hn_1 \\ \hn_2
	\end{array}\right).
\end{equation}
The functions $(\kappa_1,\kappa_2)$ define the so-called \emph{normal development} of the curve $r_0$. If the Frenet-Serret frame exists, then
\begin{equation}\label{eq:kappaRelationship}
	\kappa_1(s) = \kappa(s) \cos(\gamma(s) - \gamma^* ),\qquad 
	\kappa_2(s) = \kappa(s) \sin(\gamma(s) - \gamma^*),
\end{equation}
where $\gamma^*$ is the angle between $\hn_1(0)$ and $\hn(0)$, and $\gamma(s)$ the \emph{integral torsion}
\begin{equation} \label{eq:gamma}
	\gamma(s) = \int_0^s \tau(s) ds.
\end{equation}
	
As we mentioned above, rotation minimizing coordinates have another prominent advantage: 
the induced metric is diagonal, unlike that of the Frenet-Serret frame. 
Indeed, if $g_e$ is the Euclidean metric on $\R^3$ and $\pi_{FS}$ is the embedding \eqref{eq:EuclidMapping}, then the induced metric $g = \pi_{FS}^*g_e$ for the Frenet-Serret frame is \cite{Mercier64}
\begin{equation}\label{eq:FrenetMetric}
	\begin{aligned}
		g &= \left( h_s^2 + \tau^2(x^2 + y^2)  \right) ds^2 +2\tau (x ds dy -y ds dx) + dx^2 + dy^2, \\
		h_s &\equiv 1-\kappa x.
	\end{aligned}
\end{equation}
This metric has non-diagonal terms due to the torsion. By contrast, for the rotation minimizing frame and the embedding, $\pi_{B}:D^2\times\R \to \R^3$  
\begin{equation}\label{pibDefine}
	\pi_{B}(\tx,\ty,s) = r_0(s) + \tx \hn_1 +\ty \hn_2,
\end{equation} 
the induced metric $\tilde{g} = \pi_B^*g_e$ becomes
\begin{equation} \label{eq:RotMinMetric}
	\begin{aligned}
		\tilde{g} &=  \tilde{h}_s^2  ds^2 + d\tx^2 + d\ty^2,\qquad \\
		h_s &\equiv 1-\kappa_1 \tx -\kappa_2 \ty,
	\end{aligned}
\end{equation}
which is now diagonal.	

Rotation minimizing coordinates are not without their drawbacks. The frame is not necessarily periodic in $s$, even for a periodic $r_0$. Hence, it must be ensured that functions, forms or vectors defined on $\tilde{M} := D^2 \times\R$ are periodic when pulled back to $M$. If $\gamma_0$ is defined as the positively oriented angle between $\hn_1(0)$ and $\hn_1(\Tee) $, and $\gamma(s)$ is 
\textit{any} function satisfying
\begin{equation}\label{eq:gammaCondition}
	\gamma(\Tee + s) - \gamma(s) = \gamma_0,
\end{equation}
then this periodicity condition is equivalent to ensuring any object is well-defined under the push-forward by $\pi_S\circ R_\gamma$ where  $\pi_S (x,y,s) = (x,y,s \mod \Tee)$ is the natural projection from the cover 
$\tilde{M}$ to $ D^2\times \bS^1$, and $R_\gamma$ is a positive rotation in the plane normal to 
$\hat{t}_0(s)$ by $\gamma(s)$ for each $s$. Note that if $\gamma$ is specifically the integral 
torsion \eqref{eq:gamma} then we will push-forward to the Frenet-Serret frame; however, 
$\gamma$ can be any function satisfying \eqref{eq:gammaCondition} and we will push-forward 
to some orthonormal periodic frame of $r_0$.
	
One other drawback of the rotation minimizing coordinates is that, unlike $\kappa$ and $\tau$, the quantities $\kappa_1$ and $\kappa_2$ do not uniquely define the curve $r_0$. However, it is clear from \eqref{eq:kappaRelationship} that the normal development of a Frenet-Serret curve is unique up to rotation (essentially up to the constant $\gamma_0$ in \eqref{eq:gammaCondition}).
	
Another trick that we will find useful is to think of $D^2\subset \C$ and use the complex coordinate $z= x+ i y$, so that the metric $g$ \eqref{eq:FrenetMetric} becomes
\begin{align*}
		g &= (h_s^2 + \tau^2 z\bz ) ds^2 + i\tau(zd\bz ds - \bz dz ds) + dzd\bz , \\
		h_s &= 1-\tfrac12 \kappa (z+\bz)
\end{align*} 
Setting the initial phase $\gamma^*$ \eqref{eq:kappaRelationship} to zero, the rotation minimizing coordinates $(\tx,\ty)$ then become
\begin{equation}\label{eq:uDefine}
	u:= \tx+i\ty =e^{i\gamma}z,
\end{equation}
so that $\tilde{g}$ \eqref{eq:RotMinMetric} is now
\begin{equation}\label{eq:metrics}
	\begin{aligned}
		\tilde{g} &= h_s^2ds^2 + du d\bu , \\
		h_s &= 1-\tfrac12(\bar{\kappa}_u u + \kappa_u\bu),\qquad \kappa_u = \kappa_1 + i \kappa_2.
	\end{aligned}
\end{equation}
Note, that even though we use this complex notation, all physical functions will be taken to be real-valued.

Under the transformation to complex coordinates 
$(\tx,\ty)\mapsto(\tx+ i \ty,\tx-i \ty) = (u,\bu)$ on $T\tilde{M}$, the basis vectors $\partial_\tx,\partial_{\ty}$ of the tangent bundle and $d\tx,d\ty$  of the cotangent bundle
push forward to
\begin{equation}\label{eq:basisVectors}
	\begin{split}
		\partial_\tx = \partial_u+\partial_\bu,&\qquad \partial_\ty = i\partial_u - i\partial_\bu,\\
		d\tx = \tfrac12(du +  d\bu), &\qquad d\ty =\tfrac{1}{2i}(du - d\bu).		
	\end{split}
\end{equation}
It follows that an arbitrary vector field  becomes
\[
	B^s \partial_s+B^\tx\partial_\tx+B^\ty\partial_\ty =  
	B^s\partial_s + B^u\partial_u+\bar{B}^u\partial_\bu ,
\]
where $B^u = B^\tx+i B^\ty$. Similarly, an arbitrary one-form becomes
\begin{equation}\label{eq:formBasis}
	a = a_s ds + a_\tx d\tx + a_\ty d\ty = a_s ds + a_u du + \bar{a}_u d\bu ,
\end{equation}
with $a_u=\tfrac12(a_\tx - ia_\ty)$.
For the case of the vector potential,  $\alpha$, \eqref{eq:alpha} gives the covariant representation 
\begin{equation}\label{eq:covariant}
	\star d \alpha = B^\flat  = B_s ds + B_u du + \bar{B}_u d\bu.
\end{equation}
Note that for the metric \eqref{eq:metrics}, these components are related to the contravariant ones by 
\[
 	B_s = h_s^{2} B^s,\quad B_u = \tfrac12 \bar{B}^u, \quad \bar{B}_u = \tfrac12 B^u.
\]

%%%%%%%%%%%%
%%%%% Floquet
%%%%%%%%%%%%	
\section{Near-Axis Hamiltonians, Floquet theory, and Normal forms}\label{sec:FloquetAndNormalForm}
	In this section we establish the Hamiltonian nature of magnetic fields near an axis, opening the study of magnetic fields to the tools of Hamiltonian mechanics. We then describe two such useful tools: Floquet theory and normal form theory. Both of these are useful in finding simple coordinates in the neighborhood of a magnetic axis, and we will use them in \cref{sec:ApplicationsToMagneticFields} to construct the ``simplest'' coordinates near a magnetic axis.
	
%%%%%%%%%%%
%%%% Hamiltonian
%%%%%%%%%%%%
\subsection{Hamiltonian near a magnetic axis}
As is well known, the dynamics of the field lines in a neighborhood of $r_0$ can be described by a non-autonomous Hamiltonian system, see Ch. 9 of \citeInline{Hazeltine03}.

%%%%%%
\begin{thm}\label{thm:HamiltonianFormulation}
	There is a tubular neighborhood $U\cong  D^2\times \bS^1$ of $r_0$ with coordinates $(x,y,s)\in U$ such that the closed orbit becomes $r_0(s) = (0,0,s)$ and there is a Hamiltonian $H: U \to \R$ such that
\begin{align*}
	\alpha &= y dx - H(x,y,s) ds\\
	\beta &= dy\wedge dx - dH \wedge ds.
\end{align*}
	That is, the one-form $\alpha$ is the \emph{Liouville one-form} of a \emph{non-autonomous Hamiltonian function $H$}. Moreover at the magnetic axis, $d_\perp H|_{r_0} = 0$.
\end{thm}

%%%%%%
\begin{proof}
%	There is a more complex proof using the coisotropic embedding of the presymplectic manifold $(M,\beta)$. This proof is neat as it showed the result is essentially the generalized Darboux theorem. Anyway, here is a more direct proof.

Take some orthonormal frame at each point on $r_0$ to define coordinates in a tubular neighborhood $(\tx,\ty,s) \in U \cong D^2\times \bS^1$ of $r_0$ such that $r_0(s) = (0,0,s)$. In such a neighborhood, the fact that $\beta$ is closed implies that it is exact, that is, there exists $\alpha$ such that $\beta = d\alpha$. Using the gauge freedom of $\alpha$ we can assume
\[ 
	\alpha = \alpha_{\tx}(\tx,\ty,s) d\tx + \alpha_s(\tx,\ty,s) ds,\qquad 
	\partial_s \alpha_{\tx}(0,0,s) = 0,
\]
so that 
\[
	\beta = d\alpha = \partial_{\ty} \alpha_{\tx} d\ty\wedge d\tx + 
	(\partial_s\alpha_{\tx} - \partial_{\tx} \alpha_s)ds \wedge d\tx + 
	\partial_{\ty} \alpha_s d\ty\wedge ds .
\]
	
The magnetic field is tangent to the axis, so  $B|_{r_0} = B_0(s)\partial_s$, where
$B_0(s) \neq 0$ by assumption. Moreover, the volume form in $(\tx,\ty,s)$ has the form $\Omega = \rho(\tx,\ty,s) d\tx\wedge d \ty\wedge ds$ for some nonzero density $\rho$. Therefore, since $\beta = \iota_B\Omega$,
we know $\beta|_{r_0} = -\rho(0,0,s) B_0(s) d\ty\wedge d\tx$ and it follows that $\partial_\ty \alpha_\tx|_{r_0} = -\rho(0,0,s) B_0 \neq 0$, $-\partial_\tx \alpha_s(0,0,s) = 0$ and $-\partial_\ty \alpha_s(0,0,s) = 0$.

Choose new coordinates $(x,y,s) = (\tx, \alpha_{\tx}, s)$. This is a diffeomorphism, locally in $(\tx,\ty)$, for all $s$ by the inverse function theorem. In these new coordinates define $H = -\alpha_s$, and then $\alpha = y dx - H ds$ and $\beta = dy\wedge dx - dH\wedge ds$ as desired.

Note that $d_\perp H|_{r_0} = -\partial_x \alpha_s(0,0,s) dx -\partial_y \alpha_s(0,0,s) dy = 0$ by the assumed form of $\beta$ on $r_0$.
\end{proof}
%%%%%

In the language of vector calculus, \cref{thm:HamiltonianFormulation} is equivalent to showing that there are coordinates such that the contravariant representation of $B$ is 
\[
	B = \nabla y\times \nabla x - \nabla H(x,y,s)\times \nabla s.
\]

%%%%%%%%%
%%%%% Floquet Theory
%%%%%%%%%%
\subsection{Normal Forms: Set-up}\label{sec:normalFormsSetUp}

Birkhoff's normal form theory seeks a choice of canonical coordinates near a periodic orbit, or fixed point, for which the Hamiltonian takes its ``simplest'' form. The definition of ``simplest'' is perhaps a matter of taste; for the Birkhoff normal form, the goal is to have as few terms as possible in the series expansion of $H$. The normal form will be the result of an iterative construction of a new coordinate system. 

A review of normal form theory is given in \cref{sec:NormalFormAppendix}. Here, we will outline the core details for the normal form near a periodic orbit or magnetic axis, $r_0$, such that $d_\perp H|_{r_0} = 0$, where $d_\perp$ is derivative perpendicular to $r_0$.
It is convenient to introduce the angle
\[
	\phi = 2\pi \frac{s}{\Tee} ,
\]
so that the axis can be thought of as a periodic orbit with period $2\pi$.

Assume that $H = H(x,y,\phi)$ is a non-autonomous Hamiltonian on $D^2\times \bS^1$ with canonical variables $(x,y)$ and such that $x=y=0$ corresponds to an isolated, $2\pi$-periodic orbit. We begin by expanding $H$ in a Taylor expansion in $x,y$
\begin{equation}\label{eq:HExpansion}
	H \sim H_0 + H_1 + \dots. 
\end{equation}
Here we denote the \textit{lowest degree} terms by $H_0$, i.e., we assume that there is a $k \in \N$ such that $H_0$ is a degree $k$ polynomial in $x,y$. Similarly, $H_i$ denotes a degree $k+i$ polynomial in $(x,y)$. All of these coefficients are $2\pi$-periodic in $\phi$. 
	
Generally, $H$ begins with quadratic terms so that $k=2$. If it does not, then the orbit $r_0$ is said to be \emph{degenerate}. Such cases can still be treated by normal form theory, however, it is much more difficult to deduce the final normal form of $H$ (see \cref{sec:NormalFormAppendix} for further details). Henceforth, assume that $k=2$.

%%%%%%%%%%
%%%%% Floquet Theory
%%%%%%%%%%	
\subsection{Floquet Theory}	
We will first ignore the higher order terms and treat the dynamics of the quadratic Hamiltonian $H_0(x,y,\phi)$ using Floquet theory. The resulting linear system is
\begin{equation}\label{eq:HamSys}
    \displaystyle
    \frac{d}{d\phi} \left(\begin{array}{c}
    				x\\ y
    				\end{array}\right)
    = J \nabla H_0 
    = A(\phi)\left(\begin{array}{c}
    			x \\ y
    		 \end{array}\right),
    \quad J =  \begin{pmatrix}  0 & 1 \\ -1 & 0 \end{pmatrix},
\end{equation}
where the matrix $A(\phi)$ is a $2\pi$-periodic Hamiltonian matrix, i.e., $JA = (JA)^T$.

Since this is a linear, time-periodic system, the core result of Floquet theory \cite{Meiss17a} applies:

\begin{thm}[Floquet-Lyapunov]\label{thm:floquet}
The fundamental matrix solution $X(\phi)$ of
\begin{equation} \label{eq:FloquetDiffEqn}
	\dot{X} = A(\phi)X,\quad X(0) = I,
\end{equation}
is of the form 
\[ 
	X(\phi) = P(\phi) e^{C \phi},
\]
where the matrix $P(\phi)$ is symplectic and $2\pi$-periodic and $C$ is a constant Hamiltonian matrix.
Moreover, $P(\phi)$ and $C$ can be assumed to be real by letting $P(\phi)$ be $4\pi$-periodic if necessary.
\end{thm}

As noted, one can take $P(\phi)$ to be a symplectic matrix whenever \eqref{eq:FloquetDiffEqn} is Hamiltonian (Thm.~3.4.2 of \citeInline{meyerIntroductionHamiltonianDynamical2009}), i.e., $P^TJP = J$. In this case, $C$ must be a Hamiltonian matrix.

The eigenvalues of $C$ 
are called the \textit{Floquet exponents}. Taking coordinates $w \in\R^2$ via $(x,y)^T = P(\phi) w$ transforms \eqref{eq:FloquetDiffEqn} to the autonomous system $\dot{w} = C w$. Consequently, in the new coordinates 
$H_0 = \tfrac12 w^T C w$ is autonomous. 
	
For a one and a half degree-of-freedom Hamiltonian system, there are two Floquet exponents, $\omega_1,\omega_2$,
which must satisfy  $\omega_1 = -\omega_2$. Thus they are either purely imaginary, purely real or both zero.

When the exponents are purely imaginary, say $\pm i\,\iotab_0$, with \textit{rotational transform}
$\iotab_0\in\R\setminus\{0\}$,
then the linear system \eqref{eq:FloquetDiffEqn} is stable. More precisely,
solutions to \eqref{eq:FloquetDiffEqn} lie on invariant tori with elliptical cross
sections on $\phi=const$ surfaces. 
It is always possible in the this case to take $P(\phi)$ to be $2\pi$-periodic.

In contrast, when the exponents are purely real, say $\pm \nu_0$ with \textit{expansivity}
$\nu_0 \in\R\setminus\{0\}$, equation \eqref{eq:FloquetDiffEqn} is hyperbolic and
the periodic orbit has invariant stable and
unstable manifolds. For so-called \emph{reflection hyperbolic} orbits, the matrix $P(\phi)$ must be taken $4\pi$-periodic. Geometrically, these orbits have stable manifolds that make a $(2j+1)\pi$ rotation as $\phi$ goes from $0$ to $2\pi$, for some $j\in\Z$. In contrast,
$P(\phi)$ can be taken $2\pi$ periodic for \emph{direct hyperbolic} orbits, which have stable manifolds that make a $2j\pi$ rotation. 
These invariant manifolds serve as separatrices for $H_0$. 

The full, nonlinear system  still have bounded solutions when the axis is hyperbolic;
however, for this to be the case there must be another magnetic axis that is elliptic. 
For example, in \cref{fig:magAxisExam} the three points on the separatrix are hyperbolic orbits, while the remaining four are elliptic orbits, and the overall system still has bounded orbits.

Finally, the Floquet exponents may vanish, and then the axis is \emph{degenerate}. More generally an elliptic case could be said to be degenerate, or \emph{resonant}, when $\iotab \in \Q$. Even though a resonant axis is linearly elliptic, higher order terms may destroy the tori of the quadratic part. 

%%%%%%%%%%%%%
%%%% Higher Order
%%%%%%%%%%%%%
\subsection{Normal Forms: Higher order}
	
	Returning to the normal form procedure, we will assume that the Floquet transformation has been made so that $H_0$ does not depend on $\phi$ and that the Floquet exponents are nonzero, so the axis is linearly elliptic or hyperbolic.
	
	We seek coordinates in a neighborhood of $r_0$ that transform $H$ to its ``simplest'' form, so that the core aspects of the dynamics can easily be understood. The most concise way to state the normal form theorem is to use the Poisson bracket; if $f,g\in C^\infty(M)$ and we have canonical coordinates $x,y$ normal to $r_0$, then the \emph{Poisson bracket} is defined as
\begin{equation}\label{eq:PoissonBracket}
	\{f,g\} =   (\partial_x f)(\partial_y g)-(\partial_y f)(\partial_x g) .
\end{equation}
To simplify the calculations for the elliptic case, we will use the complex conjugate variables $(u,\bu)$, with $u = x+ iy$. In these coordinates the Poisson bracket becomes
\begin{equation}\label{eq:PoissonBracketComplex}
	\{f,g\} = 2i [(\partial_\bu f)(\partial_u g) - (\partial_u f) (\partial_\bu g) ]  .
\end{equation}
Note that even with the complex coordinates all physical functions are real-valued.
	
The following theorem gives the desired normal form for $H$.
	
%%%%%%%%
\begin{thm}\label{thm:NormalForm}
Let $H$ be a Hamiltonian system containing a linearly elliptic or hyperbolic periodic orbit $r_0$ of period $2\pi$. There exists a formal, canonical, $2\pi$-periodic (possibly $4\pi$-periodic), near-identity, change of variables $w = \Phi(u,\phi)$ that transforms the Hamiltonian \eqref{eq:HExpansion} to
\[
	\tilde{H} \sim \sum \tilde{H}_j(w,\phi) ,
\]
such that
\begin{equation}\label{eq:cohomology}
	\{\tilde{H}_j,H_0\} + \partial_\phi \tilde{H}_j = 0, 
\end{equation}
for all $j \geq 0$.
\end{thm}
%%%%%%%%
\noindent
This theorem is due to Birkhoff \cite{Birkhoff27}, and most books on Hamiltonian mechanics contain a proof. A particularly thorough account is given in \citeInline{meyerIntroductionHamiltonianDynamical2009}. The proof is constructive, giving an iterative procedure to compute the normal form at each order $j$. Some of the details of the computation are given in \cref{sec:NormalFormAppendix}. 
	
The terms in the normal form $\tilde{H}$ depend on whether the axis is hyperbolic or elliptic and, in the latter case, resonant or not.
%%%%%%%%%	
\begin{corollary}\label{cor:NormalForm}
Let $\tilde{H}$ be a non-autonomous Hamiltonian system that contains a periodic orbit $r_0$.
Then there are local coordinates $(x,y,\phi)$ such that:
\begin{enumerate}[(i)]
\item if $r_0$ is linearly elliptic with Floquet exponents $\pm i \iotab_0$ then,
if $\iotab_0 \notin \Q$ the formal normal form becomes 
\[ 
	\tilde{H}(x,y,\phi) \sim F(x^2+y^2) ,
\]
for some function $F: \R \to \R$; by contrast, in the resonant case, $\iotab_0 = p/q \in\Q$,
\begin{equation}\label{eq:ResonantNormal}
	\tilde{H} \sim \tfrac12 \tfrac{p}{q} \rho^2 + K(\rho, q\theta + p\phi),
\end{equation}
for some function $K$ where $(\theta, \rho)$ are defined by $x+iy = \rho e^{i \theta}$; and
\item if $r_0$ is linearly hyperbolic then the formal normal form becomes
\[ 
	\tilde{H}(x,y,\phi) \sim F(x y). 
\]
\end{enumerate}
\end{corollary}

A remarkable fact about normal forms for $1\tfrac12$ degree of freedom systems is that they are always formally integrable. This is most easily seen when the axis is non-resonant ($ \iota\notin\Q$) elliptic or hyperbolic so that normal form Hamiltonian  of \cref{cor:NormalForm} is independent of $\phi$. 
Thus, the Hamiltonian $\tilde{H}$ is a formal integral of the system. 

For the resonant elliptic case, the normal form \eqref{eq:ResonantNormal} depends only on the single
angle-like variable $q\theta + p\phi$. Thus one can do a time-dependent canonical transformation to a frame
that rotates with this angle to obtain a new Hamiltonian that is autonomous \cite{meyerIntroductionHamiltonianDynamical2009}.
In these new coordinates, the lowest order term $H_0$ is removed, and the Hamiltonian begins
with terms of degree $q$.
Thus the elliptic orbit becomes a \emph{degenerate} magnetic axis.
Nevertheless, since the system is now autonomous, it is formally integrable.  An example is shown in \cref{fig:res4plt} for $p/q = 1/4$. Note that the lowest order resonant terms in this case are quartic.
The (nonlinear) stability of the axis $x=y=0$ depends, in this case, on the size of the resonant terms \cite{meyerIntroductionHamiltonianDynamical2009}.

It is, however, important to note that the integrability of the normal form is misleading since the normal form expansion is generally only formal. Indeed, the power series for the coordinate transformation $\Phi$ of \cref{thm:NormalForm} typically does not converge, even in a neighborhood of the magnetic axis. Of course, if one knows that $\Phi$ is smooth or analytic then immediately one obtains the integrability of the system. There is a partial converse; if it is known that the system is integrable and the integral is nondegenerate (in particular non-resonant), then $\Phi$ must be smooth or analytic. The proof is recalled in \cite{burby2021integrability}.

\begin{figure}[ht]
	\centering
	\includegraphics[width=\linewidth]{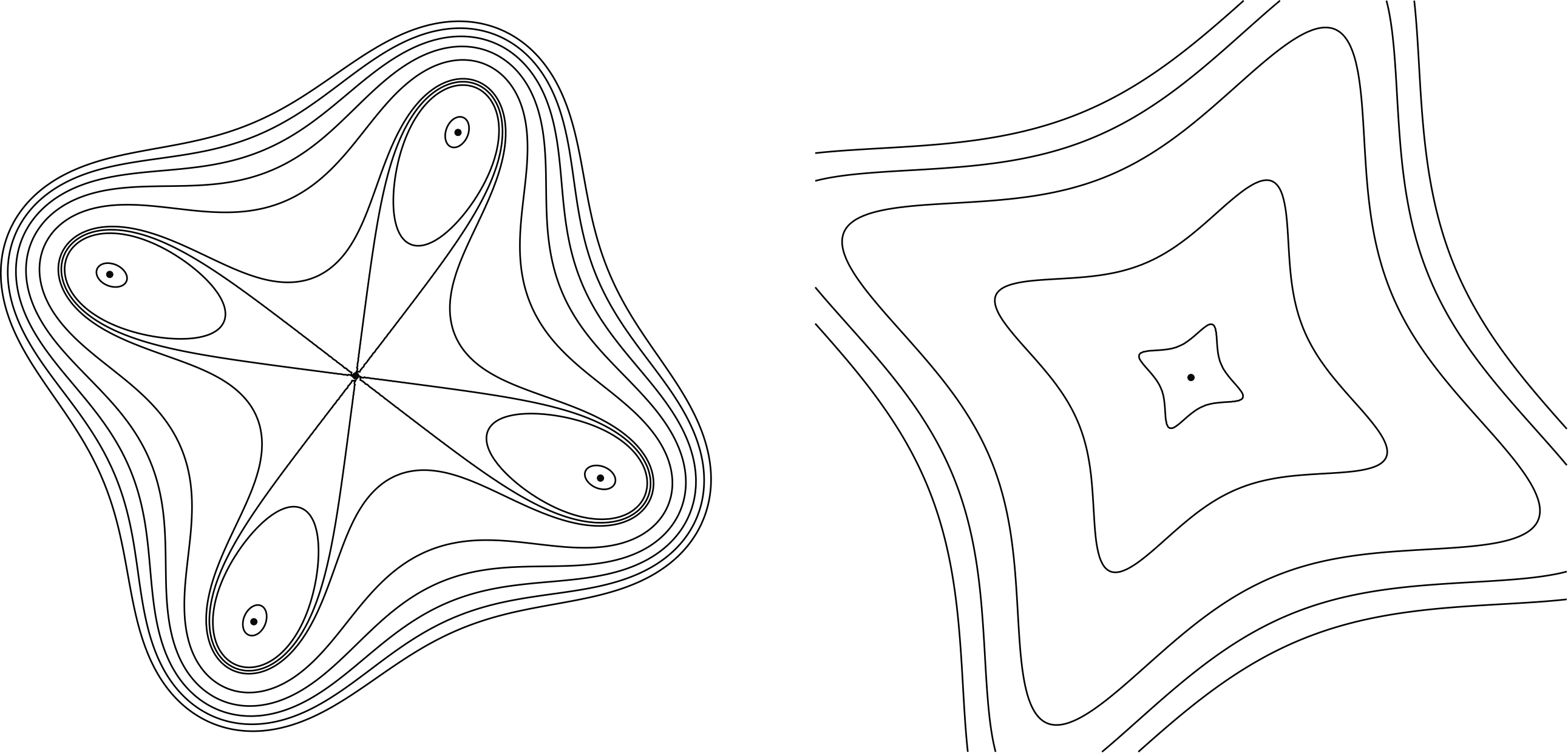}
	\caption{Two examples of an integral corresponding to a 1:4 resonance. On the left the central magnetic axis is unstable, whilst on the right it is stable.}
	\label{fig:res4plt}
\end{figure}
	
%%%%%%%%%%
%%%%% Near Resonant
%%%%%%%%%%
\subsection{Normal forms: Near Resonance}\label{sec:NearResonant}

A great benefit of understanding near-axis expansions through normal form theory is the ability to understand \emph{near} resonant phenomena. Suppose that the on-axis rotational transform $\iotab_0 = (p/q + \eps)$ for a  \emph{resonance detuning} $\eps$.
The key idea is to treat $\eps$ as formally small and to find the normal form of $H$ using the resonance $p/q$. In doing so, the normal form will be valid as $\eps$ crosses zero and may produce a better understanding of the phase space topology further away from the magnetic axis.

Of course, the rationals are dense in the reals, so there is always a $p/q$ arbitrarily close to $\iotab_0$. 
However, if $q$ is large, any resonant terms that appear will not enter the normal form until the $q^{th}$ degree terms in $x,y$. Although the following analysis will still work, it is only low-order resonances that are of primary concern near the axis.

Concretely, suppose the Hamiltonian is of the form
\[
	H = \tfrac12 (p/q + \eps) (x^2 + y^2) + \dots.
\]
When $\eps$ is formally small, it can be neglected in $H_0$ and
and the resonant normal form becomes \eqref{eq:ResonantNormal}. At this stage, we can add back the $\eps$ term
under the ordering assumption it is a small as the first resonant term, i.e., $\eps \rho^2 \sim \rho^q$.
The resulting Hamiltonian again depends only on the combination $q\theta + p\phi$,
and so it is integrable---with an invariant that can be obtained by a time-dependent transformation as before.

The topology and bifurcations of the phase portraits for different values of $q$ as $\eps$ passes through $0$ are well understood (see, for example, \citeInline{arnoldMathematicalAspectsClassical2006}). The usual consequence is a stable region about the axis followed by a $q$-island chain at a distance $\sim\sqrt{\eps}$ from the axis. However, the cases $q = 2,3,4$ are special, since the detuning term appears at an order comparable with the resonant normal form terms. A 1:3 near-resonant example was depicted in \cref{fig:magAxisExam}.

We will give an example use of this near resonance analysis in \cref{sec:Examples}.

%%%%%%%%%%%
%%%%% Application
%%%%%%%%%%
\section{Application to Magnetic Axes}\label{sec:ApplicationsToMagneticFields}
In this section we apply the classical Floquet and normal form theory to magnetic axes.
	
%%%%%%%%%%%
%%%%% Formal Hamiltonian
%%%%%%%%%%
\subsection{Formal Hamiltonians for magnetic axes}\label{sec:FormalHamiltonian}
	
Given the rotation minimizing coordinates of \cref{sec:FramingAxis}, defined in a tubular neighborhood of the magnetic axis, we now present an iterative scheme to directly compute the Hamiltonian and normal form coordinates for a given nondegenerate
magnetic axis. 

%%%%%%%%%%%%%%%%%
%% MHS
%%%%%%%%%%%%%%%%%

%%%%%%%%
%%%% Beltrami and Vacuum
%%%%%%%
\subsubsection{Series Expansions}
	
In this section we construct the canonical Hamiltonian $H$ for Beltrami and vacuum magnetic fields in the neighborhood of a magnetic axis $r_0$ by solving \eqref{eq:BeltramiAlpha} for the vector potential $\alpha$.
	
It will be convenient to solve the conditions of \eqref{eq:BeltramiAlpha} for $\ta$, the vector potential in the $(s,u)$ rotation minimizing coordinates since the metric is diagonal on $\tilde{M}$. Once this is done, we will use \eqref{pibDefine} to impose the constraint that there exists $\alpha = \pi_{B}^*\ta$, that is, that $\ta$ pushes forward to a periodic one-form  on $M$ when rotated to a periodic frame through some $\gamma(s)$ satisfying \eqref{eq:gammaCondition}.
		
For ease in computing canonical coordinates in \cref{sec:canonicalcoordinates}, it will be convenient to
use gauge freedom to choose a representation of $\ta$ different from \eqref{eq:formBasis}.

%%%%%%
\begin{lemma}\label{lem:guageChoice}
Up to gauge freedom, any real-valued one-form on $\tilde{M}$ can formally be written as
\begin{equation}\label{eq:alphaAssumption}
	\ta(u,\bu,s) = \ta_s(u,\bu,s) ds -\tfrac{1}{4 i}\ta_u(u,\bu,s) (u d\bu - \bu du),
\end{equation}
where $\ta_s,\ta_u:\tilde{M}\to\R$. Furthermore if the original form is analytic at $u=0$,
then so is $\ta$. 
\end{lemma}
%%%%%%

%%%%%%	
\begin{proof}
	An arbitrary one-form $a$ \eqref{eq:formBasis} is equivalent to \eqref{eq:alphaAssumption} under a gauge transformation if there exists a function $F:\tilde{M}\to\R$  such that $a - \ta = dF$.
In this case, necessarily  $da - d\ta = 0$. In fact, since $\tilde{M}$ is simply connected, this condition is also sufficient. Writing out each component of the condition yields,
\begin{align}
	\partial_u \ta_s      &= \tfrac{1}{4i}\bu\partial_s \ta_u + \partial_u a_s- \partial_sa_u ,
				\label{eq:firstcond}\\
	 \partial_\bu(\bu \ta_u) + \partial_{u}(u \ta_u) 
	                         &= 4i( \partial_\bu a_u - \partial_u \bar{a}_u), \label{eq:secondcond}
\end{align}
since the third, $ds \wedge d\bu$, component simply gives the complex conjugate of \eqref{eq:firstcond}.
First consider \eqref{eq:secondcond} as an equation determining a real-valued function $\ta_u$ 
given an arbitrary complex valued $a_u$.
Indeed, this can be solved at least formally about $u = 0$. To see this, expand each function 
as a power series in  $u,\bu$. Note that the operator $f \mapsto\partial_\bu(\bu f) + \partial_u (u f) $ 
maps monomials $u^{k_1} \bar{u}^{k_2} \mapsto  (2+k_1 + k_2)u^{k_1} \bar{u}^{k_2}$. 
Hence, for each monomial, we can solve the equation by simply dividing by $ (2+ k_1 + k_2)$, which is always nonzero.
	
Given such a solution to \eqref{eq:secondcond} there then exists, for each value of $s \in \R $, a function $F$ such that
\begin{equation}\label{eq:alphausol}
	a_u - \tfrac{1}{4i}\bu\ta_u = \partial_u F,\qquad 
	\bar{a}_u + \tfrac{1}{4i} u\ta_u = \partial_{\bu} F.
\end{equation}
Since $\ta_u$ is smooth in $s$ then so is $F$. 
Substituting this form into \eqref{eq:firstcond} yields
\[ 
	\partial_u \ta_s = \partial_u a_s -\partial_s \partial_u F. 
\]
Hence, taking $\ta_s =  a_s - \partial_s F$ gives a solution to \eqref{eq:firstcond}.
\end{proof}
%%%%%%

Now we use the form \eqref{eq:alphaAssumption} to solve the Beltrami equation \eqref{eq:BeltramiAlpha}.
In the metric \eqref{eq:metrics}, the covariant components \eqref{eq:covariant} become
\begin{equation}\label{eq:Bcov}
	\begin{aligned}
		%\star d \ta &= \tfrac12 h_s ( \partial_\bu (\bu\ta_u) + \partial_u (u \ta_u) ) ds \\
		%&\qquad + i h_s^{-1}\left( \partial_u \ta_s - \tfrac{1}{4i}\bu\partial_s\ta_u \right)du\\
		%	&\qquad - i h_s^{-1}\left( \partial_\bu \ta_s + \tfrac{1}{4i}u\partial_s \ta_u \right)d\bu
	B_s        & = \tfrac12 h_s ( \partial_\bu (\bu\ta_u) + \partial_u (u \ta_u) ) ,\\
	B_u        & = \frac{i}{h_s}\left( \partial_u \ta_s - \tfrac{1}{4i}\bu\partial_s\ta_u \right) , \\
	\bar{B}_u  & = - \frac{i}{h_s}\left( \partial_\bu \ta_s + \tfrac{1}{4i}u\partial_s \ta_u \right) .
	\end{aligned} 
\end{equation}
Applying the operator $\star d$ once more to obtain \eqref{eq:BeltramiAlpha} gives
\begin{equation}\label{eq:stardstard}
	\begin{aligned}
		2 i h_s ( \partial_\bu B_u - \partial_u \bar{B}_u ) &= k B_s ,\\
		i\left( \partial_u B_s - \partial_s B_u \right) &= k h_s B_u , \\
		-i\left( \partial_\bu B_s - \partial_s \bar{B}_u \right) &= k h_s \bar{B}_u .
	\end{aligned}
\end{equation}
This set, upon substitution for $B$ in terms of $\alpha$ from \eqref{eq:Bcov}, corresponds to three \textsc{pde}s for the vector potential components $\ta_s$ and $\ta_u$.
	
To formally solve \eqref{eq:stardstard} for $\alpha$ we expand each component in a series in $u$ and $\bu$, 
\begin{equation}\label{eq:alphaseries}
	\begin{aligned}
		\ta_s \sim \sum_{j=0} \ta_s^j(s,u,\bu), \\
			\ta_u \sim \sum_{j=0} \ta_u^j(s,u,\bu),
	\end{aligned}
\end{equation}
where each $\ta_\times^j$ is a degree $j$, homogeneous polynomial in $u,\bu$ with complex coefficients that are functions of $s$. Substituting the series expansion \eqref{eq:alphaseries} into the Beltrami condition \eqref{eq:stardstard} then gives
\begin{subequations}\label{eq:BeltramiSeriesEqns}
	\begin{align}
		\partial_{u\bu} \ta_s^{n} &= \tfrac{1}{4}\left\{  \im\left(\partial_\bu 
			(\bu \partial_s \ta_u) + \tfrac12 h_s^{-1} \kappa_u\bu\partial_s \ta_u\right)
			-2\re\left(\kappa_u \partial_{u}\ta_s \right) - k B_s\right\}_{n-2} ,
			\label{eq:BeltramiSeriesEqns_as} \\
		L_u\ta_u^{n} &= \left\{ \tfrac12 h_s^{-2}\bar{\kappa}_u B_s + h_s^{-1} \partial_s B_u 
			- i k B_u \right\}_{n-1} , \label{eq:BeltramiSeriesEqns_au} \\
		\bar{L}_u\ta_u^{n} &= \left\{\tfrac12 h_s^{-2}\kappa_u B_s + h_s^{-1} \partial_s \bar{B}_u 
			+i k \bar{B}_u \right\}_{n-1} , \label{eq:BeltramiSeriesEqns_aub}
		\end{align}
\end{subequations}
where $L_u$ is defined by
\[
	L_u \ta =: \tfrac{1}{2} \partial_u (\partial_\bu \bu \ta_u +\partial_u u \ta_u) ,
\] 
and $\bar{L}_u$ is the equivalent under $u \leftrightarrow \bu$.
The braces $\{\cdot\}_j$ in \eqref{eq:BeltramiSeriesEqns} denote the $j^{th}$ order term from the 
formal series \eqref{eq:alphaseries}. 
The right hand sides of \eqref{eq:BeltramiSeriesEqns} depend on
the components of $\alpha$ to at most $n-1$.
As a consequence, the equations can be solved
iteratively. We formulate this as a proposition.
	
%%%%%
\begin{proposition}\label{prop:FormalSeriesSolution}
For any smooth $\gamma(s)$ satisfying \eqref{eq:gammaCondition} there is a formal solution to \eqref{eq:BeltramiSeriesEqns} of the form
\begin{equation}\label{eq:alphasols}
	\begin{aligned}
		\ta_s^{n}(u,\bu,s) &= A_n(s)z^n + \bar{A}_n(s)\bz^n + 
		R_n(z,\bz,s) , \\
	\ta_u^{n}(u,\bu,s) &= \alpha_z^{n}(z,\bz,s),
	\end{aligned}
\end{equation}
where
\begin{equation}\label{eq:zDefine}
	z = e^{-i\gamma}u ,
\end{equation}
and $R^n$ and $\alpha_z^n$ are real, degree-$n$ homogeneous polynomials in $z,\bz$
with coefficients periodic in $s$ and dependent on $\ta_s^{k},\ta_u^k$ for 
$k<n$, and each $A_n$ is a free, complex valued function. 
In particular, if each $A_n$ is taken $T$-periodic in $s$ then the formal series 
$\ta_s,\ta_z$ are $T$-periodic in $s$.
	
Moreover, by subjecting \eqref{eq:BeltramiAlpha} to the additional constraint 
$\star d \alpha|_{r_0} = B_0 ds$, we have $ \ta_u^0 = \alpha_z^0 =  B_0(s)$, and we can choose
$\alpha_s^0=\alpha_s^1=0$ without changing $B$.
\end{proposition}

\begin{proof}		
We prove the proposition by induction on the degree in \eqref{eq:alphasols}. 
As the right hand side of \eqref{eq:BeltramiSeriesEqns} vanishes for $n = 0$,
and for $n=1$ for (a), it follows  that $\alpha_z^0, \alpha_s^0$, and $\alpha_s^1$ are free functions.
Make the particular choice $\alpha_s^0 = \alpha_s^1 = 0$ and 
$\alpha_z^0 = B_0(s)$ where $B_0(s) = B_s(0,0,s) = B^s(0,0,s)$ is the magnetic field on axis.

Assume the result is true for all $k\leq n$ and consider order $n$. The right hand side of \eqref{eq:BeltramiSeriesEqns_as} --- evaluated at order $n-2$ --- depends on $\alpha_z^{j}$ for $j\leq n-2$ and $\alpha_s^j$ for $j\leq n-1$. For the second and third equations, we must know these components order $n-1$. 
As a consequence, the first equation can be solved first to obtain $\alpha_s^n$ before solving for $\alpha_z^n$.

The left hand side of \eqref{eq:BeltramiSeriesEqns_as} can be thought of as a linear operator on the vector space of real, homogeneous polynomials. Specifically, let $\mathcal{H}_n^\gamma$ be the vector space of homogeneous degree-$n$ real polynomials in $z, \bz$ with $T$-periodic coefficients. Then
\[ 
	\partial_{u\bu} : \mathcal{H}_n^\gamma \to \mathcal{H}_{n-2}^\gamma.
\]
In order to get a solution to \eqref{eq:BeltramiSeriesEqns_as}, we need to prove the right hand side 
is in the image of $\partial_{u\bu}$. Necessarily, the right hand side must be in $\mathcal{H}_{n-2}^\gamma$. 
With the assumption that the proposition is true for all $ n\leq k-1$, a calculation confirms this is indeed true. In order to see this, note that $\kappa_z = e^{-i\gamma} \kappa_u$ is periodic by \eqref{eq:metrics}; as a consequence
$h_s$ is also periodic.
	
Now we need to check the right hand side of \eqref{eq:BeltramiSeriesEqns_as} is in the image of $\partial_{u\bu}$. Note that $\dim(\mathcal{H}_n^\gamma) = n+1$.
The kernel of $\partial_{u \bu}$ acting on $\mathcal{H}_n^\gamma$ is spanned by $\re(u^{n}), \im(u^{n})$. 
Hence, as the kernel is two dimensional and $\dim(\mathcal{H}_{n-2}^\gamma) = n-1$, then by the rank-nullity theorem it is guaranteed that $\partial_{u\bu}$ is surjective. Hence, assuming the proposition for $n\leq k-1$, there is a solution to \eqref{eq:BeltramiSeriesEqns_as} for $n = k$. 
	
Now, a necessary and sufficient condition for \cref{eq:BeltramiSeriesEqns_au,eq:BeltramiSeriesEqns_aub}
is for $\div B = 0$.
Indeed, these equations correspond to the system \eqref{eq:stardstard}. By knowing $\alpha_s$ to degree $n$,
and $\alpha_z$ to degree $n-1$, we can obtain $B_u$ to order $n$. 
We want to show that the final two equations are consistent for $B_s$.
In turn this will show that \eqref{eq:BeltramiSeriesEqns} is consistent. 
A quick rearrangement yields
\[  
	\partial_u B_s  = \partial_s B_u - i h_s k B_u, \qquad
	\partial_\bu B_s   = \partial_s \bar{B}_u + i h_s k \bar{B}_u.
\]
Since  $\tilde{M}$ is simply connected, a necessary and sufficient condition for these equations 
to be solvable is for
\[
	\partial_\bu\left(\partial_s B_u - i h_s k B_u \right) - 
	\partial_u\left(\partial_s \bar{B}_u + i h_s k \bar{B}_u\right) = 0.  
\]
This condition is satisfied as,
\begin{align*}
		\partial_\bu\left(\partial_s B_u - i h_s k B_u \right) - 
		\partial_u\left(\partial_s \bar{B}_u + i h_s k \bar{B}_u\right) &=  
		   \partial_s\left(\partial_\bu B_u - \partial_u \bar{B}_u\right) - 
		   i k  \partial_\bu \left(h_s \bar{B}_u\right) -i k \partial_u \left(h_s \bar{B}_u\right) \\
		&= -i k \partial_s(\tfrac12 h_s^{-1} B_s ) - i k  \partial_\bu \left(h_s \bar{B}_u\right) -
		i k \partial_u \left(h_s \bar{B}_u\right) \\
		&= -2 k h_s^{-1} \div B.
\end{align*}
It follows that the equations are consistent since $\div B = 0$ by assumption. 
	
Since \eqref{eq:BeltramiSeriesEqns} is consistent, a solution at order $n=k$ will exist provided the all coefficients are known for $n \leq k-1$. The result follows by induction.
\end{proof}
	
From the preceding proposition we can obtain a formal solution to
\begin{equation}\label{eq:alphasol2}
	\ta = \alpha_s(z,\bz,s) ds - 
	\tfrac{1}{4i}\alpha_z(z,\bz,s)\left( u d\bu - \bu du \right) .
\end{equation}
Consequently, by choosing $\gamma$ that satisfies \eqref{eq:gammaCondition} and then rotating through $R_{-\gamma}$ we obtain periodic coordinates
and a primitive form $\alpha$ satisfying \eqref{eq:BeltramiAlpha} so that
\begin{equation}\label{eq:alphainz}
	\alpha = \left(\tfrac12\tau\alpha_z(z,\bz,s) z \bz + \alpha_s(z,\bz,s)\right)ds - 
	\tfrac{1}{4i}\alpha_z(z,\bz,s)\left( zd\bz - \bz dz \right),
\end{equation}
where $\tau(s) := \gamma^\prime(s)$. Note that, if the Frenet-Serret frame exists and $\gamma$ is taken to be the integral torsion \eqref{eq:gamma}, then $\tau$ is the torsion of $r_0$.
	
The explicit solution to order $n=4$ of equations \eqref{eq:BeltramiSeriesEqns} is given in \cref{sec:ExplicitExpressions}. The leading order terms will be needed in \cref{sec:NormalFormCoordinatesMagAxis} and are given by
\[
	\begin{aligned}
		\alpha_s^2 &= A_2 z^2 + \bar{A}_2 \bz^2 - \tfrac14 k B_0 z \bz , \\
		\alpha_z^0 &= B_0(s) , \\
		\alpha_z^1 &= \tfrac{1}{3} B_0 \left(\bar{\kappa}_z z +  \kappa_z\bz \right) , \\
	\end{aligned}
\]
where the $A_2$ is a complex-valued, $\Tee$-periodic function of $s$, and $\kappa_z = e^{-i\gamma(s)} \kappa_u$. Note that, if $\gamma$ can be taken as the integral torsion \eqref{eq:gamma}, we have $\kappa_z = \bar{\kappa}_z= \kappa$ the curvature and $\tau$ the torsion of $r_0$.

%%%%%%%%%%%%%%%%%%%%%%%%%%%%%%%%%%%
%% Canonical Coordinates
%%%%%%%%%%%%%%%%%%%%%%%%%%%%%%%%%%%

\subsubsection{Canonical coordinates}\label{sec:canonicalcoordinates}
	
We have now computed in \cref{prop:FormalSeriesSolution} the vector potential $\alpha$. By finding canonical coordinates, we will be able to get an expression for the Hamiltonian $H$. A method to compute canonical coordinates using series expansions could be obtained using the work of Cary and Littlejohn \cite{Cary83}. However, the canonical coordinates are easily obtained as a consequence of our choice of gauge from \cref{lem:guageChoice}.
	
In complex coordinates $Z = X + i Y$, the canonical Liouville one-form is
given, up to a closed one-form, by
\begin{equation}\label{eq:canonComplex}
	\begin{aligned}
		\alpha &= \tfrac{1}{4i}\left( Z d\bZ - \bZ dZ  \right) - H ds \\
				&= \tfrac12\left(Y dX - X dY\right) - H ds ,
	\end{aligned}
\end{equation}
where $H:M\to\R$ is the Hamiltonian. Recalling
that $\alpha_z$ is real, \eqref{eq:alphainz} can be transformed into the form \eqref{eq:canonComplex} using
\begin{equation}\label{eq:canonicalTranform}
	Z = \bar{z}\sqrt{\alpha_z(s,z,\bz)}.
\end{equation}
and $H$ becomes
\begin{equation}\label{eq:Hamiltonian}
	H = -\tfrac12 \tau(s) Z \bar{Z} - \alpha_s.
\end{equation} 
The first four orders of $H$ are explicitly given in \cref{sec:ExplicitExpressions}. The leading order is
\begin{equation}\label{eq:H0}
H_0 	= -\bar{A}_2 B_0^{-1} Z^2 - A_2 B_0^{-1} \bZ^2 + \left(\tfrac14 k-\tfrac12 \tau \right) Z\bZ \\ 
\end{equation}
	
%%%%%%%%%%%%%%%%%%%%%%%%%%%%%%%%%%%
%% Recovering the Magnetic Field
%%%%%%%%%%%%%%%%%%%%%%%%%%%%%%%%%%%	

\subsubsection{Recovering the magnetic field}
We have found canonical coordinates $(Z,\bZ)$ which put the one-form $\alpha$ into the form \eqref{eq:canonComplex}. The beauty of this form is that almost all of the information about the magnetic vector field is stored in the single function $H$. In order to extract the magnetic field from $H$, we can first use the fact that $\iota_B\Omega = \beta = d\alpha$. Explicitly, we have that 
$\Omega = \rho  dZ\wedge d\bZ \wedge ds$ where $\rho = \sqrt{g_Z}$ is the density from the metric $g_Z$ induced by the transformation to canonical coordinates. Then
\begin{align*}
	\iota_B\Omega &= \rho B_Z d\bZ\wedge ds - \rho B_{\bZ} dZ\wedge ds 
		+ \rho B_s dZ\wedge d\bZ, \\
	 d\alpha &=  \tfrac{1}{2i} dZ\wedge d\bZ - dH \wedge ds. 
\end{align*}
It follows that
\[
	B_Z = -\rho^{-1} \partial_{\bZ} H,\qquad 
	B_{\bZ} = \rho^{-1} \partial_{Z} H,\qquad
	B_s = \tfrac{1}{2i} \rho^{-1}.
\]
Note that these are the Euler-Lagrange equations scaled by $\tfrac{1}{2i} \rho^{-1}$.
	
The density $\rho$ is a complicated expression, even when computed as a power series in $Z,\bZ$. However, to compute the normal form, we will not be working directly with $B$, but instead with the Hamiltonian vector field for $H$,
\[ 
	B_H = 2i \rho B = -2i \partial_{\bZ} H \partial_Z + 2i \partial_Z H \partial_{\bZ} + \partial_s .
\]
Consequently, the complication of computing $\rho$ is bypassed.

%%%%%%%%%
%%%%% Normal Form coordinates
%%%%%%%%%
\subsection{Normal form coordinates near magnetic axis}\label{sec:NormalFormCoordinatesMagAxis}
	
%%%%%%%
%%%% Rotational Transform
%%%%%%%
%\subsubsection{Computing the rotational transform and expansivity}

Normal form theory will be useful in the current analysis of near-axis
expansions as it will allow for the elimination of as many terms as possible
in the series expansion of $H(s,Z,\bZ)$. These removable terms are dependent
on the quadratic component $H_0$ given explicitly in \eqref{eq:H0}.
Specifically, as was highlighted by \cref{cor:NormalForm}, if the axis is elliptic we need to compute the on-axis rotational transform $\iotab_0$, or, if the axis is hyperbolic, we need the on-axis expansivity $\nu_0$.

It is possible to get an explicit formula for these constants in terms of geometric quantities, such as the total torsion and curvature of the magnetic axis.
First, note the leading order dynamics from $H_0$ is given by
\begin{equation}\label{eq:leadingorderdynamics}
	\dot{Z} = -2i\partial_{\bZ} H_0 = 4i \bar{C}_2(s)\bZ - i\left(\tfrac12 k - \tau(s)\right)Z ,
\end{equation}
with $C_2(s) := B_0(s)^{-1}\bar{A}_2(s)$.
Theorem \ref{thm:floquet} guarantees a canonical transformation 
$w = F(s,Z,\bZ)$ that is linear in $Z, \bZ$ bringing \eqref{eq:leadingorderdynamics} into
\begin{equation}\label{eq:veqns}
	\dot{w} = i \iotab_0 w,
\end{equation}
in the elliptic case and
\begin{equation}
	\dot{w} = \nu_0 \bar{w}.
\end{equation}
in the hyperbolic case.

In order to obtain $F$ explicitly, it is useful to use geometry.
In the elliptic case, the transformation $F$ takes invariant tori of the linearized dynamics, 
that are in a tubular neighborhood $U$ of the axis, into other invariant tori in 
$ D^2\times \bS^1$ that are given simply by level sets of $w \bar{w} $ with $w\in D^2 $. 
Similarly, in the hyperbolic case, $F$ maps invariant surfaces of the linearized dynamics, that are in a tubular neighborhood of $U$ of the axis, to invariant surfaces that are level sets
of $w^2 + \bar{w}^2$ in $ D^2 \times \bS^1$.

In either case, this transformation can be broken down into a rotation
\begin{equation}\label{eq:deltaRotation}
	Z = e^{i\delta(s)}V,
\end{equation}
that aligns the principal axes of each ellipse or hyperbola
with the coordinate axes. Here $\delta(\Tee+s) - \delta(t) = 2\pi j$  in
the elliptic or direct hyperbolic case, and $\delta(t+\Tee) - \delta(t)= (2j+1)\pi$ 
in the reversed hyperbolic case,  for some $j\in \Z$.
This is followed by a canonical scaling 
\begin{equation}\label{eq:etaScaling}
	V = \cosh\eta(s) v - \sinh\eta(s)\bar{v},
\end{equation}
for $\Tee$-periodic $\eta(s)$, in order to scale each ellipse or hyperbola to 
obtain symmetry between the major and minor axes.

Applying the transformations to \eqref{eq:leadingorderdynamics} yields
\begin{equation}
	\begin{aligned}
	\dot{v} &= -\tfrac{i}2 \left( 4\left( C_2e^{2 i \delta} 
		+ \bar{C}_2 e^{-2i\delta} \right)\sinh 2\eta 
		+ ( k + 2\delta^\prime - 2\tau)\cosh 2 \eta \right) v \\
		&\qquad + i \left( 4 e^{-2i\delta} \bar{C}_2 \cosh^2\eta  
		+ 4e^{2i\delta}C_2 \sinh^2\eta + \left(\tfrac12 k + \delta^\prime - \tau \right)\sinh 2\eta 
		- i \eta^\prime \right)\bv .
	\end{aligned}
\end{equation}

In order for ellipses to be invariant we require $\frac{d}{ds}v\bv = 0$. Computing this condition yields
\begin{equation}\label{eq:etadiffeq}
	C_2(s)\cosh 2\eta = -\tfrac{1}{4}e^{-2i\delta}\left(  \left(\tfrac12 k + \delta^\prime -\tau \right)\sinh 2 \eta + i\eta^\prime\cosh 2\eta\right),
\end{equation}

In order for hyperbolas to be invariant we require $\frac{d}{ds}(v^2 + \bv^2) = 0$.
Computing this condition yields
\begin{equation}
	\label{eq:etadiffeq_hyp}
	C_2(s) \sinh 2\eta = -\tfrac{1}{4}e^{-2i\delta}\left(  \left(\tfrac12 k + \delta^\prime - \tau \right)\cosh 2 \eta + i\eta^\prime \sinh 2\eta \right).
\end{equation}
Note that there is an additional constraint here; whenever $\eta = 0$ we must have $\delta^\prime = \tfrac12 k + \tau$. In fact it must be asserted that $\delta^\prime = \tfrac12 k + \tau + F(s)\sinh 2\eta$ for some function $F(s)$.

Finding the functions $\delta,\eta$ from \eqref{eq:etadiffeq} or \eqref{eq:etadiffeq_hyp} given the
function $C_2(s)$ can not, in general, be done analytically. 
Indeed this would amount to solving the general Floquet problem for \eqref{eq:leadingorderdynamics}. 
However, when designing useful magnetic fields, it is
perhaps easier to work backwards by choosing $\delta,\eta$ and
letting them determine the free function $C_2$.

Applying the two transformations to the leading order Liouville one-form $\alpha^0$ yields
\begin{equation}\label{eq:alphav}
	\alpha^0 = \tfrac{1}{4i}\left( v d\bv - \bv dv \right) - H_0 ds ,
\end{equation}
where the Hamiltonian has transformed to either
\[
	H_0 = \tfrac12\left( \tfrac12 k + \delta^\prime -\tau \right)\sech (2\eta) v \bv,
\]
for the elliptic case or
\[
	H_0 = \tfrac14\left(\tfrac12 k + \delta^\prime -\tau \right) \csch(2\eta) (v^2+\bv^2) .
\]
for the hyperbolic case.
	
In the elliptic case, the dynamics are now given by
orbits contained in invariant tori $\tfrac12v\bv = const$ with frequency
of poloidal rotation
\begin{equation}\label{eq:iotaDefinition}
	\iota(s) = \int_0^s (\tfrac12 k + \delta^\prime - \tau)\sech 2\eta ds.
\end{equation}

To complete the transformation $F$ we require that the rotation rate is
constant along the toroidal angle $s$. This can be done by averaging through
the nonautonomous canonical transformation
\begin{equation}\label{eq:iotaRotation}
	v = e^{-i \left(\iota(s) - \iotab_0 s \right)}w,\qquad 
	\iotab_0 := \iota(\Tee)/\Tee,
\end{equation} 
to transform the Hamiltonian to
\[
	H_0 = \tfrac{1}{2}\iotab_0 w \bw ,
\]
as desired.
	
In the hyperbolic case, the dynamics have been reduced to orbits contained in invariant level sets of
$\tfrac14(v^2 + \bv^2) = const$ with a contraction rate of
\begin{equation}\label{eq:nudefinition}
	\nu(s) = \int_0^s (\tfrac12 k + \delta^\prime - \tau ) \csch(2 \eta) ds.
\end{equation}
To complete the transformation $F$ we require that this contraction rate be constant along the
toroidal angle $s$. This can be done through a canonical rescaling using the average contraction rate.
Specifically, make the canonical transformation
\begin{equation}
	v = \cosh(\nu(s) -  \nu_0 s)w - i\sinh (\nu(s) -  \nu_0 s)\bw, \quad \nu_0 := \tfrac{1}{\Tee}\nu(\Tee) .
\end{equation} 
This transforms the Hamiltonian to
\[
	H_0 = \tfrac{1}{4}\nu_0 (w^2+\bw^2) ds,
\]
as desired. 

We summarize the results of this subsection as the following lemma.
%%%%%
\begin{lemma}\label{lem:FloquetTransformations}
The Liouville one-form can be transformed to
\begin{equation}\label{eq:wHamiltonian}
\begin{aligned}
		\alpha &= \tfrac{1}{4i}\left( w d\bw - \bw dw \right) - H_0 ds - \sum_{n=1} H_n(s,w,\bw) ds, \\
		H_n &= C_n(s) w^{n+2} + \bar{C}_n(s)\bw^{n+2} + P_n(s,w,\bw),
\end{aligned}
\end{equation}
where $P_n(s,w,\bw)$ is a polynomial in $w,\bw$ of homogeneous degree $n+2$
with coefficients that are $\Tee$ periodic functions in $s$ and dependent on
$C_j,j\leq n$.

Specifically there are functions $\delta$ and $\eta$ such that when the axis is elliptic 
the transformation, using \eqref{eq:iotaDefinition}, is
\[ 
	Z = e^{i\delta}(\cosh(\eta)v - \sinh(\eta) \bar{v}),\qquad
    v = e^{-i(\iota(s)-\iotab_0)}w,  
\]
and $H_0 = \tfrac{1}{2}\iotab_0 w \bw$

Similarly, when the axis is hyperbolic the transformation, using \eqref{eq:nudefinition}, is 
\[
	Z = e^{i\delta}(\cosh(\eta) v - \sinh(\eta) \bar{v}) ,\quad  
	v = \cosh(\nu(s) -  \nu_0 s)w - i\sinh (\nu(s) -  \nu_0 s)\bw ,
\]
and $H_0 = \tfrac{1}{4}\nu_0(w^2+\bw^2)$.
\end{lemma}
%%%%%%%%

%%%%%%%%%%%
%%%%% Normal Form
%%%%%%%%%%%
%\subsubsection{Normal Form}
Once the leading order terms of the Hamiltonian $H_0$ are in the 
autonomous form given in \cref{lem:FloquetTransformations}, we are in a position to
determine the Birkhoff normal form for $H$ using \cref{thm:NormalForm} from \cref{sec:FloquetAndNormalForm}.
The result will, of course, depend upon whether the axis is elliptic or hyperbolic, and in
the elliptic case, on whether $\iotab_0$ is rational or not. An example of
a near-resonant normal form for the elliptic case is given in \cref{sec:Examples}.

%%%%%%%%%%%
%%%% Constant Curvature
%%%%%%%%%%
\section{Example: Curves with Constant Torsion}\label{sec:Examples}

In this section normal form theory is applied to an example curve for a magnetic axis. 
For simplicity, a choice of $r_0$, and the free functions is made to ensure
the Hamiltonian has only finitely many harmonics at each order. In \cref{sec:ConstantCurvatureCurves},
a family of closed curves with simple curvature and torsion functions is described.
Then, in \cref{sec:NormalFormExamples}, a near-axis expansion for these curves is made,
the corresponding Hamiltonian function computed, and the normal form analyzed.

\subsection{Closed curves of constant torsion}\label{sec:ConstantCurvatureCurves}
To obtain a simple form for the magnetic axis, we follow Karcher \cite{karcherClosedConstantCurvature2020}
who obtained a simple, discrete symmetry condition under which a Frenet curve, $r(s)$, is a closed loop.\footnote
%%%%
{More generally a necessary condition is that the Floquet matrix of the system \eqref{eq:FrenetSerret}
is the identity \cite{Hwang81}, but this is a harder computation.}
%%%%
To start, assume that the curvature and torsion functions are even, i.e., there is a point $s_j$ such that
\begin{equation}\label{eq:symProperty}
	\kappa(s_j +s) = \kappa(s_j - s), \qquad \tau(s_j +s) = \tau(s_j - s) .
\end{equation}
%thus $\tau$ and $\kappa$ are even functions about $s_j$ for each $j$. 
Such a symmetry implies a corresponding symmetry of the resulting space curve:
	
%%%%
\begin{lemma}\label{lem:DiscreteSym}
	If $r(s)$ is a Frenet space curve \eqref{eq:FrenetSerret} satisfying \eqref{eq:symProperty} for some $s_j \in \R$, then it is symmetric at $s_j$ with respect to a rotation by $\pi$ about the curvature normal $\hn(s_j)$.
\end{lemma} 
%%%%
%%%%
\begin{proof}
Without loss of generality, suppose that $s_j =0$.
For any smooth functions $\kappa(s)$ and $\tau(s)$, there exists
a unique curve $r(s)$, up to Euclidean transformations. 
Denote by $r_+(s)$ the curve with curvature $\kappa_+(s) = \kappa(s)$ and torsion $\tau_+(s) = \tau(s)$
and by $r_-(s)$ the curve with curvature $\kappa_-(s) = \kappa(-s)$ and torsion $\tau_-(s) = \tau(-s)$.
Assuming that both satisfy the same initial conditions,
$r_\pm(0) = 0$, $\hn_\pm(0) =(1,0,0)$, then uniqueness implies that $r_-(s) = r_+(-s)$. It follows that
\[
	\hat{t}_-(0) = -\hat{t}_+(0), \quad \hat{n}_-(0) = \hat{n}_+(0), \quad \hat{b}_-(0) = -\hat{b}_+(0). 
\]
Thus, since $\kappa_+ = \kappa_-$ and $\tau_+ = \tau_-$, and $r_+(0) = r_(0)$
the curves $r_\pm(s)$ must be identical up to a 
$\pi$-rotation about the shared normal $\hn_\pm(0)$, denoted by $R_\pi$. It follows that $ r_+(s) = R_\pi \cdot r_-(s) = R_\pi \cdot r_+(-s) $. 

More generally when $s_j \neq 0$, this symmetry holds for $r(s) = r_+(s-s_j)$.
\end{proof}
%%%%

If we suppose that there are at least two of these even symmetry points, 
say $s_0 \neq s_1$, for \eqref{eq:symProperty}, then there are infinitely many of them.
%%%%
\begin{lemma}\label{lem:InfiniteSymPoints}
		If a function has two even symmetry points $s_0 \neq s_{1}$ then it is periodic with period $2(s_1-s_0)$. 
		%for each $j \in \bZ$, $s_j = j(s_1 - s_0)+s_0$ is also an even symmetry point.
\end{lemma} 
%%%%
\begin{proof}
By a simple calculation,		
\begin{align*}
	f(s+2(s_1-s_0)) &= f(s_1 + (s_1-2s_0+s)) = f(s_1-s_1+2s_0-s)\\
		 &= f(2s_0-s) = f(s_0 + (s_0-s)) = f(s_0-s_0+s)\\
		 & = f(s) ,
\end{align*}
implying that $f$ is periodic with period $2(s_1-s_0)$.
\end{proof}

Thus if both $\kappa$ and $\tau$ are even about two symmetry points, then they are both periodic. Consequently,
if the curve $r(s)$ is not closed, then it has infinitely many even symmetry points. 	
More generally, if the set of symmetry points is assumed discrete, then \cref{lem:InfiniteSymPoints} immediately implies that the full set is countable.

Geometrically, the proof of \cref{lem:InfiniteSymPoints} shows that the entire curve $r(s)$ can be constructed from a fundamental segment $ \{r(s) \,|\, s\in [s_0,s_1]\}$. This is done by first rotating about $\hn(s_1)$ to get a symmetry point $s_2 = 2(s_1-s_0) + s_0$.  A similar rotation about $\hn(s_2)$ then gives $s_3$, and so on. Likewise, rotation about $\hn(s_0)$ gives $s_{-1}$ and so on. 

The discrete symmetry established in \cref{lem:DiscreteSym} can be used to find closed curves. If we can find that there exists some $m \in \N$ such that $r(s_m) = r(s_0)$ then we are guaranteed that $r(s)$ is periodic. 
The following gives a sufficient condition for periodicity.

%%%%%%%
\begin{proposition}\label{prop:closedSufficientCondition}
Suppose $r(s)$ has curvature and torsion satisfying \eqref{eq:symProperty} for at least two distinct values $s_0$ and $s_1$, and $r(s_0) \neq r(s_1)$. Let $\ell_j$ be the normal line $\{r(s_j) + t \hn(s_j)\ |\ t\in\R \}$. If $\ell_0$ and $\ell_1$ intersect and do so at an angle that is a rational multiple of $\pi$, then $r(s)$ is closed. 
	\end{proposition}
%%%%%%%
%%%%%
\begin{proof}
		From \cref{lem:DiscreteSym}, the discrete maps $R_0, R_1$ corresponding to $\pi$-rotation about respectively $\hn(s_0),\hn(s_1)$, are discrete symmetries of $r(s)$. From \cref{lem:InfiniteSymPoints}, the points $s_j = j(s_1 - s_0) + s_0$ for $j\in \N$ are also symmetry points with discrete maps $R_j$ (that are $\pi$-rotations about $\hn(s_j)$). We claim that if $\ell_0,\ell_1$ intersect at a point $p$ then all $\ell_j := \{r(s_j) + t \hn(s_j) | t\in\R \}$ intersect at $p$. Indeed, the line $\ell_1$ is the fixed set of the map $R_1$. Hence, applying $R_1$ to $\ell_0$ yields another line which will also intersect $\ell_1,\ell_0$ at $p$. But this line is precisely $\ell_2$. By induction, it must be true that all $\ell_j$ intersect at a point $p$. Moreover, since all the rotations are by $\pi$, all these lines lie on the same plane spanned by $\ell_0$ and the vector $r(s_1) - r(s_0)$.
		
		Finally, the angle of intersection between $\ell_0,\ell_1$ is the same as the angle of intersection between $\ell_j,\ell_{j+1}$. Hence, if this angle is a rational multiple of $\pi$, then the discrete mapping of $\ell_{j-1}$ into $\ell_{j+1}$ by $R_j$ will be periodic. Hence, for some $m\in \N$, $r(s_m) = r(s_0)$, implying the curve is closed.
	\end{proof}

The sufficient conditions of this proposition can be reformulated as
\begin{equation}\label{eq:closedCurveConds}
	\cos^{-1}(\hn(s_0) \cdot \hn(s_1)) = \tfrac{p}{q} \pi,\qquad 
	\left(r(s_0) - r(s_1)\right)\cdot\left(\hat{n}(s_0) \times \hat{n}(s_1)\right) = 0,
\end{equation}
for $p,q \in \N$.
Here the first condition guarantees that the angle is rational and the second that the two lines 
$\ell_0$ and $\ell_1$ both lie in a plane containing the points $r(s_0)$ and $r(s_1)$, 
so they necessarily intersect. 

As a simple example with two even symmetry points, \eqref{eq:symProperty},  
consider curves with torsion that is constant and curvature that has a single harmonic:
\begin{equation}\label{eq:SimpleFrenet}
	\tau(s) = \tau_0, \quad \kappa(s) = a_0 + a_1 \sin(s). 
\end{equation}
Clearly, $s_0 = \pi/2$, and $s_1=3\pi/2$ are even symmetry points, and by \cref{lem:InfiniteSymPoints}
so are $s_j = (\tfrac12 + j)\pi$ for each $j \in \Z$.
Approximate values of $(\tau_0,a_0,a_1)$ that give closed curves can be obtained by solving 
the Frenet-Serret equations numerically for $s\in [s_0,s_1]$, for a range of values of 
$(\tau_0,a_0,a_1)$, and applying a Newton iteration to find solutions to \eqref{eq:closedCurveConds}.
One such closed curve, corresponding to the parameter values
\begin{equation}\label{eq:ClosedParams}
	\tfrac{p}{q} = \tfrac{4}{5}, \quad (\tau_0,a_0,a_1) \approx (0.8,0.549008,0.48878) ,
\end{equation}
is shown in \cref{fig:fifthCurve}. Note that this curve has length $\Tee = 10\pi$.
	
	\begin{figure}[ht]
		\centering
		\includegraphics[width=\linewidth]{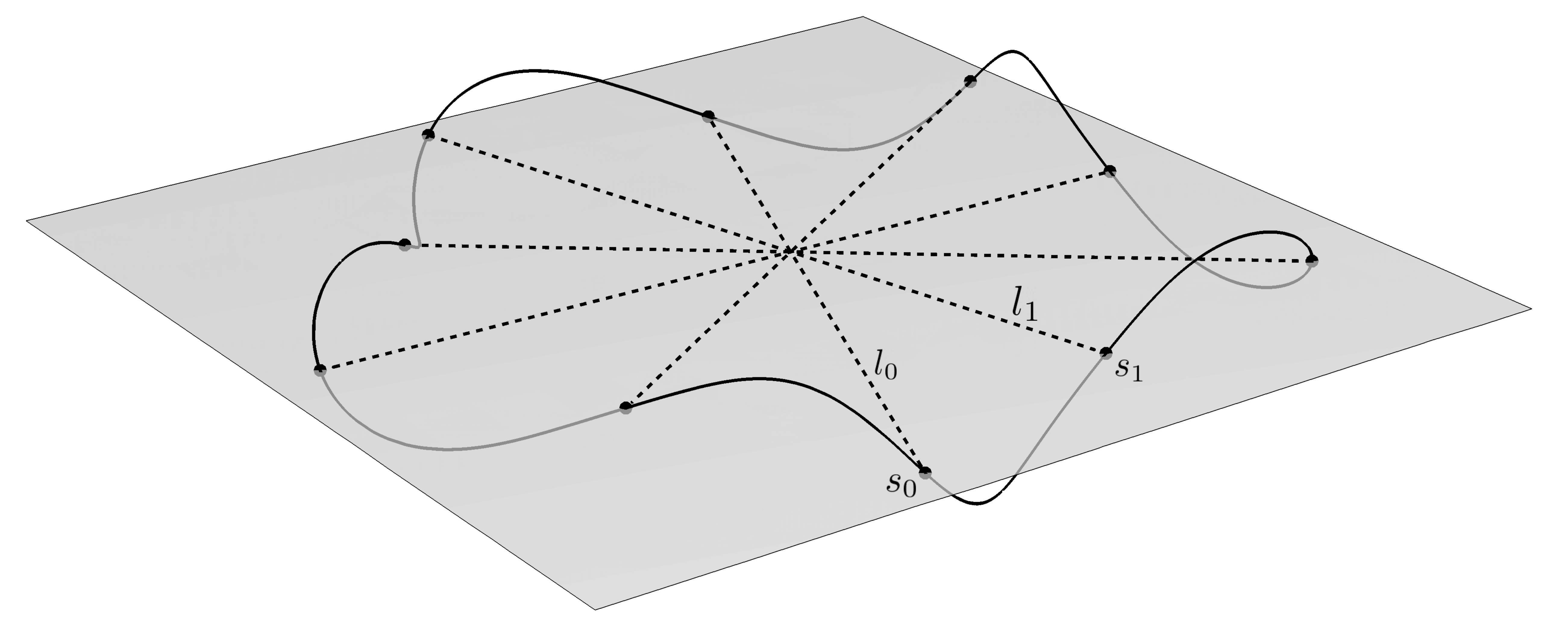}
		\caption{A closed curve satisfying \eqref{eq:closedCurveConds} with parameters \eqref{eq:ClosedParams}. 
		It has a 10-fold  symmetry.}
		\label{fig:fifthCurve}
	\end{figure}

\subsection{Normal form examples}\label{sec:NormalFormExamples}

In this section we give two examples of a normal form calculation to quartic degree in the variables $w,\bw$ of \eqref{eq:wHamiltonian}. We follow the iterative procedure outlined in \cref{sec:FloquetAndNormalForm}. The available free functions to this order, as determined in \cref{sec:FormalHamiltonian}, are summarized in \cref{tab:freeFunctions} for an elliptic axis. 

\begin{table}[ht]
	\caption{Available free functions at a given degree.}
	\label{tab:freeFunctions}
	\begin{tabular}{c|c l}
		Degree 	& Free functions & Conditions \\\hline
		2		& $k$ 		  & real\\
		2		& $B_0(s)$    & real, periodic \\
		2		& $\delta(s)$ & periodic $\mod 2\pi$ \\
		2		& $\eta(s)$   & real, periodic. \\
		3		& $A_3(s)$ 	  & complex, periodic. \\
		4		& $A_4(s)$ 	  & complex, periodic.
	\end{tabular}
\end{table}

We will make a choice of the free functions that simplifies the number of terms appearing in the
Floquet-transformed Hamiltonian (\cref{eq:wHamiltonian}) before computing the normal form. 
Moreover, we will choose the Beltrami constant $k$ so that the on-axis rotation is close to a resonance, $\iotab_0 = 1/3+\eps$. This allows for two comparative normal form examples; the first taking $\iotab_0$ so that there  are  no resonance terms up to degree four, and the second using the near resonant normal form as described in \cref{sec:NearResonant}. 
	
We will use the family $r_0(s)$ with \eqref{eq:SimpleFrenet} with the parameters \eqref{eq:ClosedParams}. At degree zero, choose $B_0(s) = 1$. At degree two, take $\eta(s) = \tfrac12 \ln 2$ and $\delta(s) = s$, see \cref{tab:near13Resonance}. This choice ensures that we can simply compute \eqref{eq:iotaDefinition} as
\begin{equation*}
	\iota(s) = \int_0^s (\tfrac12 k +\delta^\prime - \tau )\sech 2\eta ds = \tfrac45(\tfrac12 k  +1 - \tau_0)s ,
\end{equation*}
a linear function of $s$. As a consequence, the Floquet transformations
(\cref{eq:deltaRotation,eq:etaScaling,eq:iotaRotation}) will produce coefficients of the Hamiltonian \eqref{eq:wHamiltonian} that are periodic functions with only low order harmonics. Moreover, by choosing $k$ appropriately, we can ensure that $\iotab_0 = \iota(T)/T = 1/3+\eps$ for some small value of $\eps$.
	
For the degree 3 and 4 terms, there is a choice of periodic, complex valued functions $A_3(s), A_4(s)$ as given in \cref{prop:FormalSeriesSolution}. We want to ensure that there are some 1:3 resonance terms when the near-resonant normal form is computed in \cref{sec:nearResonantExample}. By \cref{cor:NormalForm}, this amounts to requiring that the coefficient of $w^3$ is non-zero. Hence, we choose $A_3(s)$ so that, after Floquet transformation, the coefficients of $w \bw^2$ and $w^2\bw$ vanish, but the coefficient of $w^3$ and $\bw^3$ do not. Similarly, we will choose $A_4(s)$ so that the coefficients of $w^3\bw, w\bw^3$ vanish. This guarantees as few coefficients at this order as possible.
	
A summary of the choice of free functions is given in \cref{tab:near13Resonance}. Through appropriate substitution of these specific free functions and consequent transformation of \eqref{eq:H0} by the Floquet transformations, the resulting Hamiltonian \eqref{eq:wHamiltonian} to degree-four is of the form,
\begin{equation}\label{eq:exampleHamiltonian}
	\begin{aligned}
		H &= H_0 + H_1^0 + H_2^0,\qquad
		H_0 =\tfrac12(\tfrac13 + \eps) w\bw , \\
		H_1^0 &= \sum_{m=-2}^2 e^{ i m s}\left(b_{3,0,m}w^3 +b_{0,3,m}\bw^3\right) ,\\
		H_2^0 &= \sum_{m=-4}^4 e^{ i m s}\left(b_{4,0,m} w^4 +b_{2,2,m} w^2\bw^2+b_{0,4,m} \bw^4\right).
	\end{aligned} 
\end{equation}
Note that $b_{j,k,m} = \bar{b}_{k,j,-m}$ since $H$ is real.
	
\begin{table}[ht]
	\caption{Choice of free functions for examples.}
	\label{tab:near13Resonance}
	\begin{tabular}{c|c l}
		Degree 	& Free functions & Choice \\\hline
		2		& $k$ 		  & $\tfrac16 (-7 + 15 \eps + 12 \tau_0)$\\
		2		& $B_0(s)$    & $ 1$ \\
		2		& $\delta(s)$ & $s$ \\
		2		& $\eta(s)$   & $\tfrac12 \ln 2$ \\
		3		& $A_3(s)$ 	  & so that coefficients of $w\bw^2, w^2\bw$ vanish \\
		4		& $A_4(s)$ 	  & so that coefficients of $w\bw^3,w^3\bw$ vanish
	\end{tabular}
\end{table}

%%%%

%%%%%%%%%
%%% Example 1
%%%%%%%%
\subsubsection{Example 1: non-resonant normal form}
Using the choice of functions given in \cref{tab:near13Resonance}, we will compute the non-resonant normal form for the Hamiltonian \eqref{eq:exampleHamiltonian}. We follow \cref{sec:TimeDependentNormal} in moving to extended phase space and using the extended bracket, \eqref{eq:ExtendedPB}, $\{\cdot, \cdot\} + \partial_s$.
The following calculation will hold provided $\iotab_0 \neq \tfrac{p}{q}$ with $q \le 4$. 

At the $d^{th}$ iteration in the normal form procedure, we need to solve the homological equation
\[ 
	\{F_d, H_0 \} + \partial_s F_d = H_d^{d-1} - H_d^{d} ,
\]
where the first term represents the degree $d$ terms after the previous $d-1$ normal form transformations have been made, and the second contains the resonant (or irremovable) terms.
That is, at each order the Hamiltonian will have the form
\[
	H_d^{d-1} = \sum_{|m| \le d} \sum_{j+k=d} {b}_{j,k,m} e^{i m s} w^j \bw^k ,
\]
with $ {b}_{j,k,m} = \bar{{b}}_{k,j,-m}$ since $H$ is real.
It is important to note that each coefficient $b_{j,k,m}$ must be computed by applying all the previous, lower order, transformations (see \cref{sec:NormalFormAppendix}).

Letting the degree-$d$ terms have the power series representation
\[ 
	F_d = \sum_{|m| \le d} \sum_{j+k=d} F_{j,k,m} e^{i m s} w^j \bw^k. 
\]
then, solving the homological equation, term-by-term gives
\begin{equation}\label{eq:FSol}
	F_{j,k,m} = -i \frac{{b}_{j,k,m}}{(k-j)\iotab_0 + m} ,
\end{equation}
providing there is no resonance: $(k-j)\iotab_0 + m \neq 0$. 
Any ``resonant'' terms are not removed, and accumulated in the normal form Hamiltonian $H_d^d$.
		
With $\eps$ taken so that $\iotab_0 \neq \tfrac{p}{q}$ with $q \le 4$, the degree $4$ Hamiltonian is given as,
\begin{equation}\label{eq:example1NormalForm}
	{H}^4 = H_0^0 + H_1^1 + H_2^2 = \tfrac12 \iotab_0 w \bw + {b}_{2,2,0} w^2 \bw^2.
\end{equation} 
Note that ${b}_{2,2,0} \in \R$ as $H$ is real.
As there are no resonant terms in the normal form \eqref{eq:example1NormalForm}, and consequently it does not depend on $s$, then $\tilde{H}$ is an approximate integral for the system. By pulling back this approximate integral by the normal form transformation, we obtain an integral in the original Floquet coordinates $w, \bw$. The resulting level sets of this function at a slice $s = 0$ is given in \cref{fig:nonresplt} for a value of $\eps = 0.01$. Moreover, some trajectories of the degree four Hamiltonian \cref{eq:exampleHamiltonian} are plotted for comparison.	
		
%%%%%%%
		\begin{figure}[ht]
			\centering
			\includegraphics[width=0.75\linewidth]{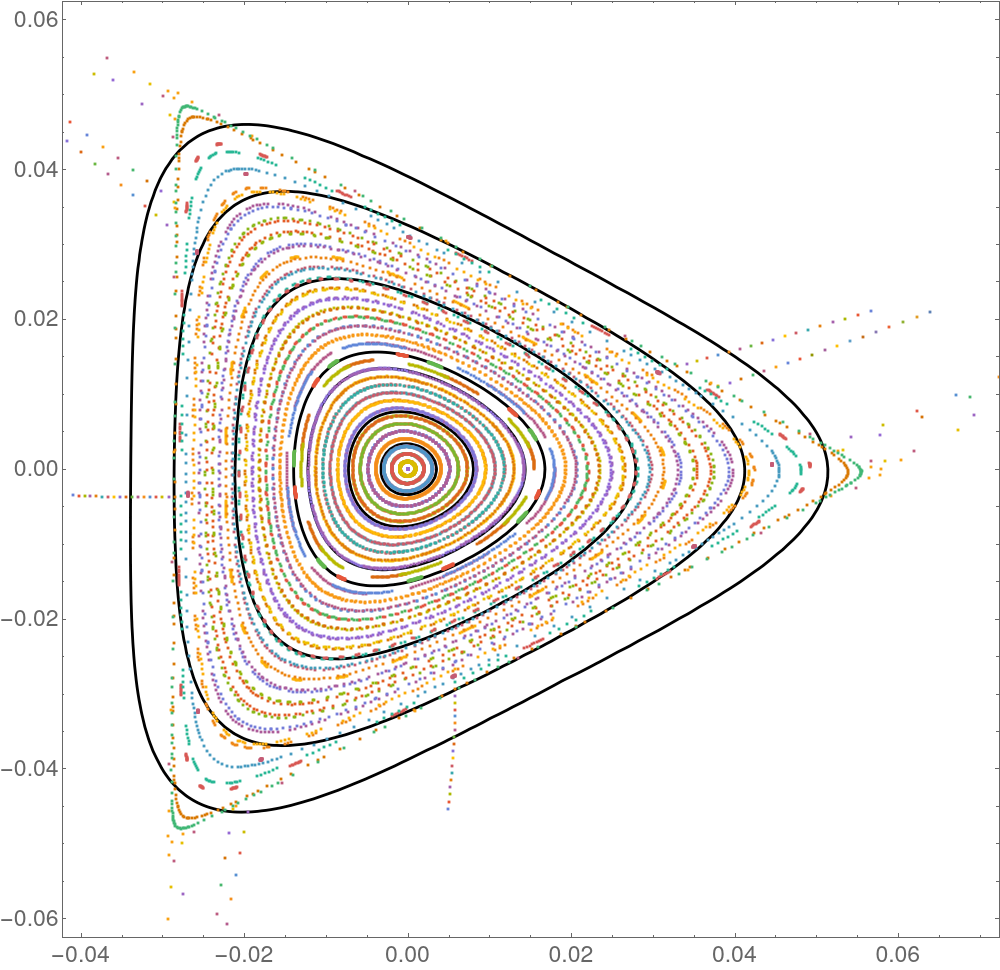}
			\caption{In black are the contour lines of the approximate invariant $
				\tilde{H}$ at a slice $s=0$ and for $\eps = 0.01$. Overlaid is a Poincar\'e plot of the orbits of the Hamiltonian system \eqref{eq:exampleHamiltonian}.}
			\label{fig:nonresplt}
		\end{figure}	
%%%%%%%%%		
	
%%%%%%%%%
%%% Near Resonant Example
%%%%%%%%
\subsubsection{Example 2: near-resonant normal form}\label{sec:nearResonantExample}
As in the previous example, we will make a choice of functions as given in \cref{tab:near13Resonance}. In contrast however, we will compute the near-resonant normal form for the Hamiltonian \eqref{eq:exampleHamiltonian}. That is we will assume $\eps$ to be small and follow the theory outlined in \cref{sec:NearResonant}.
		
As before, the general solution for the normalizing functions $F_d$ is given \cref{eq:FSol}.
As we are essentially taking $\iotab_0 = 1/3$, at degree 3 there is a resonant term
\[ 
	H_1^1 = b_{3,-1} e^{- i s}w^3 +\bar{b}_{3,-1} e^{is}\bw^3.
\]
		
The degree $4$ terms of the normal form, $H_2^2$, will only contain the irremovable term $\tilde{b}_{2,2,0}\in \R$, similar to the previous example. Here $\tilde{b}_{2,2,0}$ is the coefficient of $w\bw$ after the transformation generated by $F_1$ has been applied to $H^0$ to give $H^1$. The degree $4$ normal form becomes
		\begin{equation}\label{eq:example2NormalForm}
			{H}^4 = H_0 + H_1^1 + H_2^2 = \tfrac12 \iotab_0 w \bw + b_{3,0,-1} e^{- i s}w^3 +\bar{b}_{3,0,-1} e^{is}\bw^3 +  \tilde{b}_{2,2,0} w^2 \bw^2.
		\end{equation} 
		
The resonant terms in the normal form \eqref{eq:example2NormalForm} prevent the system from being autonomous. As a consequence, the normal form itself is not an integral. However, through a rotation, $\tilde{w} = e^{-1/3 i s} w$ we obtain the approximate integral
\begin{equation}\label{eq:example2Integral}
	J = \tfrac12 \eps w \bw + b_{3,-1} \tilde{w}^3 +\bar{b}_{3,-1} \bar{\tilde{w}}^3
		 +  b_{2,2,0} \tilde{w}^2 \bar{\tilde{w}}^2.
\end{equation} 
By pulling back this approximate integral by the rotation and the normal form transformation, we obtain an integral in the original Floquet coordinates $w, \bw$. The resulting level sets of this function at a slice $s = 0$ is given in \cref{fig:resplt} for a value of $\eps = 0.01$. Moreover, some trajectories of the degree $4$ Hamiltonian \cref{eq:exampleHamiltonian} are plotted for comparison.
		
%%%%%%%
		\begin{figure}[ht]
			\centering
			\includegraphics[width=0.75\linewidth]{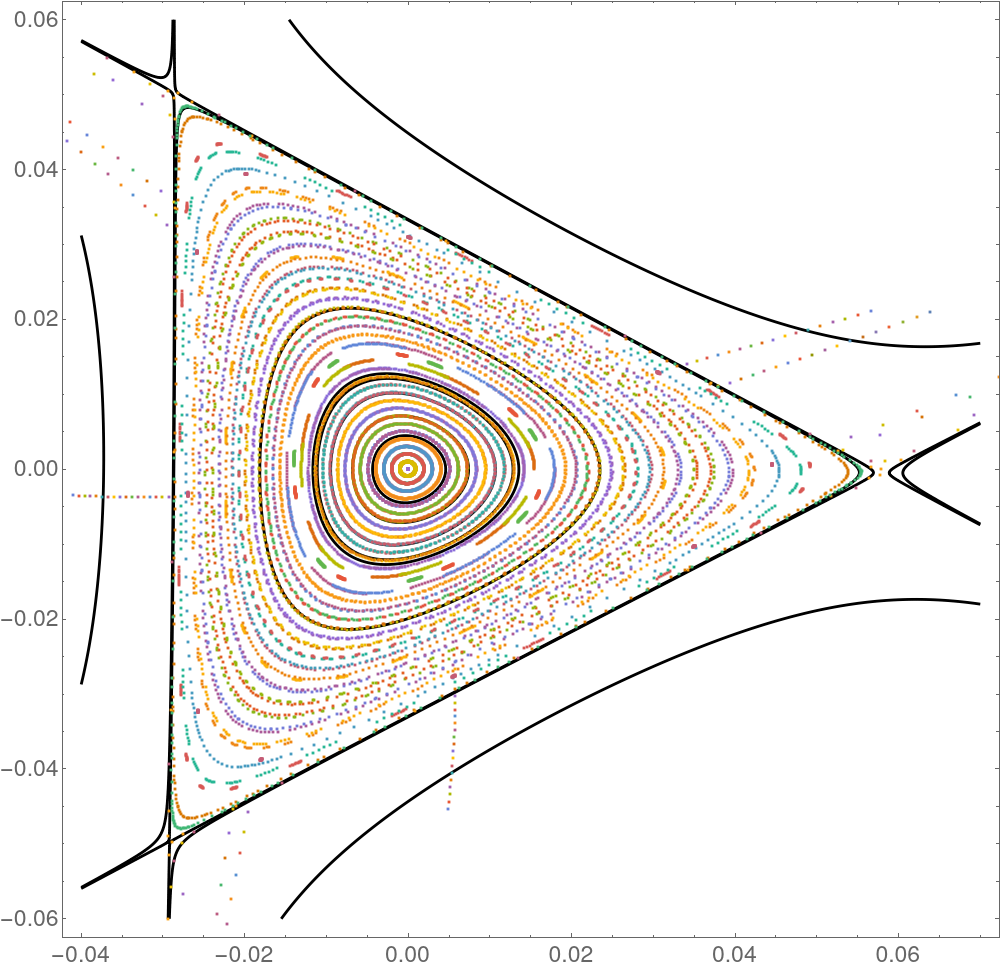}
			\caption{In black are the contour lines of the approximate invariant $J$ at a slice $s=0$ and for $\eps = 0.01$. Overlaid is a Poincar\'e plot of the orbits of the Hamiltonian system \eqref{eq:exampleHamiltonian}.}
			\label{fig:resplt}
		\end{figure}	
%%%%%%%%

%%%%%%%%%%
%%%% Concluding Remarks
%%%%%%%%%%
\section{Concluding Remarks}\label{sec:conclusion}
% brief summary
In this paper, we studied near-axis expansions for Beltrami and vacuum magnetic fields. These were 
introduced through the lens of differential forms, motivating a Hamiltonian perspective and facilitating the application of Floquet and normal form theory. Ultimately these techniques gave a way to iteratively compute simple coordinates and an approximate integral for fields near a magnetic axis. We gave two examples, the first analyzed through a regular normal form, and the second through a near-resonant normal form.

% Further exploration of the intrinsic geometry
The language of differential forms reveals how Beltrami or vacuum fields gives a manifold a contact or cosymplectic structure, respectfully. The fact that Beltrami fields form contact structures has been used previously to establish topological properties of Beltrami fields (for example see \citeInline{etnyreContactTopologyHydrodynamics2000a,encisoBeltramiFieldsExhibit2020,cardonaConstructingTuringComplete2021}). As far as we are aware, the cosymplectic structure of a vacuum field has yet to be explored. We hope that future work may further illuminate this subtle difference.

In this paper, several generalizations were made that permit the application of near-axis expansions to a wider variety of configurations. Firstly, we showed how to construct a rotation minimizing frame without a Frenet-Serret frame, using the ideas of Bishop \cite{Bishop75}. Consequently, it is possible to study axes with points of vanishing curvature while still using the traditional framework of Mericer \cite{Mercier64}. Secondly, we demonstrated how to implement a near-axis expansion for a hyperbolic axis, which may prove useful for the study of divertors \cite{Boozer15,Boozer18}. Finally, we were able to carry out these expansions without assuming the existence of flux surfaces or of non-resonance. Nevertheless, we constructed approximate integrals using the normal form. This contrasts the common claim that an existence requirement for flux surfaces is non-resonance of the axis \cite{Jorge20a}.

% How quasisymmetry, isodynamic, omnigenity in general may interact with resonant surfaces is possible.
As we demonstrated, normal form theory gives a way to understand the possible topology of magnetic fields near an axis, in particular for those near or at resonance. The resonant normal form will allow the investigation of  properties, such as omnigenity or quasisymmetry, near island chains and hyperbolic axes. This may enlarge the category of configurations with these properties.

% Comparison of near-resonant and non-resonant example. near resonant gives both a better approximation of potential flux surfaces and a good indication of existence and position of any separatrix.
Moreover, as demonstrated by comparing \cref{fig:resplt} and \cref{fig:nonresplt}, a near-resonant normal form can give a better approximation of flux surfaces as well as help locate any separatrices. 
%Future work would do well to employ the near-resonant normal form to understand critical stellarator designs such as W7X.

% Further exploration of Beltrami fields: Useful for step-pressure equilibria
Finally, this paper applies the near-axis expansion, for the first time, to Beltrami fields. Such fields have recently become the central point in computing stepped-pressure MHD equilibria that have open regions with $\nabla p = 0$ where the field is Beltrami, e.g. \cite{Hudson20}. Although Beltrami fields are generically chaotic, the existence of even approximate flux surfaces would be advantageous. Our techniques will be useful to construct fields with such flux surfaces.

%%%%%%%%%%
%%%% Appendix
%%%%%%%%%%
%\pagebreak
\begin{appendix}

%%%%%%%%%%
%%%% Normal Form Theory
%%%%%%%%%%
\section{Normal Form Theory}\label{sec:NormalFormAppendix}
In this section we will give a  proof of \cref{thm:NormalForm}. There are many proofs in the literature (see, for instance, \citeInline{meyerIntroductionHamiltonianDynamical2009}). However, because these proofs are constructive, and are used in the computations in \cref{sec:FloquetAndNormalForm} and \cref{sec:NormalFormExamples}, it is of use to be more explicit.
	
As outlined in \cref{sec:FloquetAndNormalForm}, normal form theory seeks a choice of canonical coordinates near a periodic orbit for which the Hamiltonian takes its ``simplest'' form. We want a way to easily apply canonical transformations to the Hamiltonian, and study how this transformation modifies the Hamiltonian at each degree in the Taylor series expansion of $H$ about the periodic orbit. As pioneered by \citeInline{dragtLieSeriesInvariant1976}, one relatively easy way to do this is through the method of Lie series.
	
This method takes advantage of the fact that the set of diffeomorphisms $\Diff(M)$ on 
a manifold $M$ is a Lie group under composition. 
The corresponding Lie algebra $\chi(M)$ is the vector space of complete vector fields on $M$. 
This gives an efficient method for computing the action of a flow of a Hamiltonian vector field
on a function.
Explicitly if $\vphi_{X}^t, \epsilon\in\R,$ is the time-$t$ flow of a vector field 
$X\in\chi(M)$, then the action of $\vphi_{X}^t$ on a function $f$ is given by the \emph{Lie series}
\begin{equation}\label{eq:LiePullBack}
	\vphi_{X}^{t*} f = \exp(t X)f := \sum_{k=0}^\infty \tfrac{t^k}{k!}(X^k f) ,
\end{equation}
where the vector field $X = X^i \partial_{i}$ can be thought of as the usual Lie derivative 
operator that acts on functions. 

For the case of Hamiltonian vector fields, the flow will be a symplectic transformation. These can be generated using the Poisson bracket $\{\cdot,\cdot\}$ on the manifold $M$, recall \eqref{eq:PoissonBracket}.
Given $H \in C^\infty(M, \R)$, the Hamiltonian vector field is $X_H = \{\cdot,H\}$,
and the Lie derivative operator becomes
\[
	L_H := \{\cdot, H \} = X_H.
\]
Moreover, it can be shown that the Lie bracket for Hamiltonian vector fields is given by
\[
	[X_G,X_F] = X_{\{F,G\}} .
\]
The point is that the map $F \mapsto \{\cdot, F\}$ is a Lie algebra homomorphism.

Thus, given a Hamiltonian function $F:M\to\R$, there is a
Hamiltonian vector field $X_F = \{\cdot,F\}$, which in turn, can be
used to generate a symplectic transformation $ \vphi_{X_F}^t$, the time $t$ flow of $X_F$. 
Consequently, the action on a function $H$ by can be computed using Lie series through
\eqref{eq:LiePullBack}, that is,
\[ 
	\vphi_{X_F}^{t*} H = \exp( L_F )H.
\]
Through this relation, we never have to deal directly with the flow $\vphi_{X_F}$, or even the vector field $X_F$, we can simply use Hamiltonian $F$ and the Poisson Bracket operator $L_F$, to compute the transformation of $H$.
	
We will proceed by first recalling the normal form transformation for an autonomous Hamiltonian, before extending it to the non-autonomous case.
	
%%%%%%%%%%
%%%%% Autonomous Normal Forms
%%%%%%%%%
\subsection{Time Independent} \label{sec:time-independent} 
Assume that $H$ is independent of time and has an equilibrium, without loss of generality, at the origin.
Let $(x,y)\in T^*\R^n$ be local canonical coordinates and $\{\cdot,\cdot\}$ the canonical Poisson bracket \eqref{eq:PoissonBracket}.
	
We review here the iterative procedure to transform
the Hamiltonian $H$ to a normal form and produce the required transformation. To start this procedure, we expand $H$ into a Taylor series about the origin 
\[
	H = H_0 + \epsilon H_1^0 + \epsilon^2 H_2^0 + \dots ,
\] 
denoting the degree of the homogeneous component $H_j$ using a subscript.
We will make the non-degeneracy assumption that $H_0$ is degree two, and so that a lower index $j$ 
indicates a degree $2+j$ polynomial in $(x,y)$. The upper index---on the higher order terms---will denote the step in the normal form procedure. We omit this for $H_0$, as it will stay fixed. 
The $\epsilon$ is introduced purely for bookkeeping.
	
For the general normalization step, consider a $j+2$ Hamiltonian $\epsilon^jF_j$ 
with degree $j+1$ vector field $\epsilon^j X_{F_j}$.
Since the $\epsilon^j$ just scales time, the corresponding ``time-one" flow becomes the symplectic transformation $\vphi_{X_{j}}^{\epsilon^j}$.
Using \cref{eq:LiePullBack} this transforms the Hamiltonian $H$ to a new form $\tilde{H}$, given by,
\begin{align*}
	\tilde{H} &= \exp(\epsilon^{j}L_{F_j}) (H_0 + \epsilon H_1^0 + \dots) \\
%	&= \exp(\epsilon^{j}L_{H_j})H_0 +\exp(\epsilon^{j}L_{H_j})H_1 + \dots + \exp(\epsilon^{j}L_{H_j})\epsilon^jH_{j}^0 + \dots \\
	&= H_0 + \epsilon H_1^0 + \dots + \epsilon^j\left( L_{F_j}H_0 + H_{j}^0 \right) + O(\epsilon^{j+1}).
\end{align*}
Note that the lowest-order effect of this transformation is to 
transform the degree $j+2$ homogeneous component of $H$, namely $H_j^0$. 
The corresponding equation at this order is
\[	
	H^1_j = L_{F_j}H_0 + H_j^0 = \{H_0,F_j \} + H_j^0.
\]

This equation, when thought of as an equation to determine the desired $F_j$:
\begin{equation}\label{eq:Homological}
	L_{H_0}F_j = \{F_j,H_0 \} = H_j^0 - H_j^1,
\end{equation}
is referred to as the \textit{homological equation}.  
The central object of study is now the linear operator $L_{H_0} = \{\cdot,H_0\}$.
Letting $\H_j$ be the vector space of degree-$j$ homogeneous polynomials in $(x,y)$, then,
since $H_0$ is quadratic,  
\[
	L_{H_0}:\H_{j} \to \H_{j}.
\]

Ideally $F_j$ would be chosen so that $H_j^1 = 0$, and the resulting Hamiltonian would
have no order-$j$ terms. However, this is only possible if 	$H_j^0 \in \Ima(L_{H_0})$. 
However, if $L_{H_0}$ is not onto, then there can be components of $H_j^0$ not in $\Ima L_{H_0}$. 
For the case that the linearized matrix $D J\nabla H_0|_{q=p=0}$ is diagonalizable (recall \eqref{eq:HamSys}), then so is the operator $L_{H_0}$\cite{meyerIntroductionHamiltonianDynamical2009}. 
Under this assumption, it is always possible to write
\[
	\H_{j} = \Ima (L_{H_0} ) \oplus \Ker(L_{H_0}).
\]
Then if we choose $H_j^1$ to be the projection of $H_j^0$ onto $\Ker(L_{H_0})$, the homological equation
\eqref{eq:Homological} can be solved for $F_j$. Of course, then the terms $H_j^1$, the ``resonant terms,''
remain in the transformed Hamiltonian $H^1$.

The normal form procedure thus begins by diagonalizing the matrix associated with the quadratic $H_0$. 
To normalize the cubic terms, $H_1^0$, we act on $H$ by the flow generated by a degree-three Hamiltonian $F_1$, or equivalently its degree-two vector field $X_1$. The transformed Hamiltonian becomes
\[ 
	 H^1  = \exp(\epsilon L_{F_1})H = H_0 + \epsilon H_1^1 + \epsilon^2 H_2^1 + \dots.  
\]

The next step in the iterative procedure is to normalize the order two terms, using a transformation generated by a Hamiltonian $\epsilon^2 F_2$. The result with be 
\[
	H^2 = H_0 + \epsilon H_1^1 + \epsilon^2 H_2^2 + \dots,
\]
where the order zero and order one terms remain unchanged.
To do this, we must solve the homological equation
\[ 
	L_{H_0}F_2 = H_2^1 - H_2^2.
\]
Again we choose $H_2^2$ so that $H_2^1-H_2^2 \in \Ima(L_{H_0})$ to compute $F_2$.

Continuing in this fashion, we can obtain the $j^{th}$ order Hamiltonian by iteration of $F_k$ for $k=1,\dots,j$, namely
\[ 
	H^j = H_0 + \epsilon H_1^1 + \dots + \epsilon^j H_j^j + \ldots.
\]
The transformation $\Phi_j$ bringing $H$ to $j^{th}$ order normal form is given by the Lie series
\[
	\Phi_j = \exp(\epsilon^j L_{F_j})\exp(\epsilon^{j-1} L_{F_{j-1}})\cdots\exp(\epsilon L_{F_1}).
\]
This can be computed as a series expansion. Note that the inverse transformation is easily computed as
\[
	\Phi_j^{-1} = \exp(-\epsilon L_{F_1})\cdots \exp(-\epsilon^{j-1} L_{F_{j-1}})\exp(-\epsilon^j L_{F_j}). 
\]

%%%%%%%%
%%%% Time Periodic
%%%%%%%%
\subsection{Time Periodic}\label{sec:TimeDependentNormal}

We will now present an outline of the time dependent normal form.
Consider a Hamiltonian that depends periodically on time $t$ so that
$H:M\times S^1\to\R$. As before, we assume that we are given canonical
coordinates $(x,y)\in T^*\R^n$ near an orbit $r_0: S^1 \to M$ 
that has period $\Tee$. We assume that coordinates (e.g. using Floquet theory) have been chosen so that
$r_0(t) = (0,0)$.

Now, decompose $H$ into its various homogeneous terms by expanding in a Taylor
series in $(x,y)$:
\begin{equation} \label{eq:HamHomogeneousComponents}
	H(x,y,t) = H_0(x,y,t) + \epsilon H_1(x,y,t) + \dots .
	\end{equation}
As before, $\epsilon$ is introduced purely for bookkeeping.

We assume that the Hamiltonian is in Floquet coordinates so that $H_0 = H_0(x,y)$ is independent of time $t$ (see \cref{sec:FloquetAndNormalForm}) and is quadratic.
We would now like to do the same as in the autonomous case, namely, simplify the system
by a near identity, canonical coordinate transformation. Unfortunately, the time dependence prevents us from
directly using the Lie derivatives $L_{H_j}$. However, this problem can be circumnavigated
by moving to extended phase space.

A point in extended phase space $(x,y,t,E) \in M \times S^1 \times \R$ has the energy variable $E$ as its fourth coordinate. The Hamiltonian is now
\[
	\tilde{H} = \tilde{H}_0 + \epsilon H_1 + \ldots, \qquad  \tilde{H}_0 := H_0+ E,
\]
and the Poisson bracket becomes
\[
	\widetilde{\{\cdot,\cdot\}} = \{\cdot,\cdot\} + \{\cdot,\cdot\}_E,\qquad
	\{F,G\}_E:= \partial_t F \partial_E G - \partial_E F \partial_t G.
\]
The normal form transformations do not need to change the time coordinate, so we consider vector fields 
$X = \widetilde{\{\cdot,F\}}$.
that are zero in the time-direction. As a consequence, the corresponding Hamiltonian, $F$, 
in the extended phase space is independent of $E$, so that
\begin{equation}\label{eq:ExtendedPB}
\begin{split}
	\tilde{L}_{H_0} &:= \widetilde{\{F,\tilde{H}_0\}} = \{F,H_0\} + \{F, E\}_E\\
		            &= \{F,H_0\} + \partial_t F .
\end{split}
\end{equation}

As before, the transformation at order $j$ corresponds to a Hamiltonian $\epsilon^j F_j$ that is degree $j+2$ in $(x,y)$, but now has $\Tee$ periodic coefficients. This gives a
degree $j+1$ vector field $\epsilon^j X_j$, and generates the transformation 
$\vphi_{X_j}^{\epsilon^j}(x,y,t,E)$.

Now, from the computation in \cref{sec:time-independent}, application of the
symplectic transformation $\vphi$ produces the homological equation
\[
	\tilde{L}_{H_0}F_j = H_j^{j-1} - H_j^j . %\widetilde{\{F_j,\tilde{H}_0\}}.
\]
It follows that, in the time dependent case, the appropriate linear operator is now $\tilde{L}_{H_0}$,
and it is $\ker{\tilde{L}_{H_0}}$ that determines the resonant terms in the normal form.

The normal form procedure is carried out in \cref{sec:FloquetAndNormalForm}, following the autonomous case. The only 
difference is the modification of the homological equation and that $\tilde{L}_{H_0}$ must be used in computing the Lie series of \cref{eq:LiePullBack}.

%%%%%%%%%%%
%%%%%Explicit Forms
%%%%%%%%%%%
\section{Explicit Expansions for the Vector Potential and Hamiltonian}\label{sec:ExplicitExpressions}

The explicit solution to the recursive equations \eqref{eq:BeltramiSeriesEqns} to order $n=4$ is given as
\begin{align*}
	\alpha_s^2 &= A_2 z^2 + \bar{A}_2 \bz^2 - \tfrac14 k B_0 z \bz ,\\
	\alpha_s^3 &= A_3 z^3 + \bar{A}_3 \bz^3 + z \bz \left( R_{2,1} z + \bar{R}_{2,1} \bz  \right) ,\\
	\alpha_s^4 &= A_4 z^4 + \bar{A}_4 \bz^4 + R_{2,2}z^2 \bz^2 + z\bz\left( R_{3,1} z^2 + 
	\bar{R}_{3,1}\bz^2 \right) ,\\
	\alpha_z^0 &= B_0(s) ,\\
	\alpha_z^1 &= \tfrac{1}{3} B_0 \left(\bar{\kappa}_z z +  \kappa_z\bz \right) ,\\
	\alpha_z^2 &=  \tfrac14 \re\left(\left( B_0\kappa_z^2 + 4 A_2(k+2\tau) 
	+ 4 i A_2^\prime \right)z^2\right) - 
	\tfrac18\left( B_0\left( k^2 - 2|\kappa_z|^2 \right) + B_0^{\prime\prime} \right) z \bz ,\\
	\alpha_z^3 &= Q_{3,0} z^3 + \bar{Q}_{3,0} \bz^3 + z\bz\left( Q_{2,1} z + \bar{Q}_{2,1} \bz \right) ,
\end{align*}
where
\begin{align*}
	R_{2,1} &= -\tfrac{1}{4} A_2 \kappa_z +\tfrac{1}{96} \left(5 i \bar{\kappa}_z  B_0^\prime+
	B_0 \left(\bar{\kappa}_z  \left(3 k+2 \tau\right)+2 i \bar{\kappa}_z^\prime\right)\right) ,\\
	%%%%%%%%%%%%%%%%%%%%%%
	R_{2,2} &= -\tfrac{1}{16} \re\left( A_2 \kappa_z ^2\right)+
	\tfrac{1}{128} \left(2 k B_0^{\prime\prime}+B_0 \left(2 k^3+|\kappa_z| ^2 \left(k-2 \tau\right) 
	- 2\im\left(\bar{\kappa}_z\kappa_z^\prime\right)\right)\right) ,\\
	%%%%%%%%%%%%%%%%%%%%%%%
	R_{3,1} &= \tfrac{1}{48} A_2 \left(-3 |\kappa_z| ^2-4 \left(k^2+\tau \left(k-2 \tau\right)-
	i \tau^\prime\right)\right)-\tfrac{1}{4} A_3 \kappa_z - \tfrac{1}{24} A_2^{\prime\prime}, \\
	&\qquad +\tfrac{1}{384} \bar{\kappa}_z  \left(9 i \bar{\kappa}_z  B_0^\prime+B_0 \left(\bar{\kappa}_z  \left(3 k+10 \tau\right)+ 10 i \bar{\kappa}_z^\prime\right)\right) -\tfrac{1}{24} i (k - 4 \tau) \bar{A}_2^\prime , \\
	%%%%%%%%%%%%%%%%%%%%%%%%%
	Q_{3,0} &= \tfrac{1}{15} A_2 \left(3 \bar{\kappa}_z  \left(k+4 \tau\right)+2 i \bar{\kappa}_z^\prime\right)+
	\tfrac{2}{5} A_3 \left(k+3 \tau\right)+\tfrac{1}{20} B_0 \bar{\kappa}_z^3 + \tfrac13 i\bar{\kappa}_z A_2^\prime + \tfrac25 i A_3^\prime ,\\
	%%%%%%%%%%%%%%%%%%%%%%%%%%%%%
	Q_{2,1} &= \tfrac{1}{10} A_2 \left(\kappa+z  \left(k+5 \tau\right)+i \kappa_z^\prime\right)+  \tfrac{3}{10}i \kappa_z A_2^\prime + 
	\tfrac{1}{80} \left(-9 \bar{\kappa}_z  B_0^{\prime\prime}+i \bar{\kappa}_z  B_0^\prime \left(2 k+7 \tau\right) \right. \\
	&\qquad \left.-
	7 B_0^\prime \bar{\kappa}_z^\prime+B_0 \left(\bar{\kappa}_z  \left(-3 k^2+k \tau+2 i\tau^{\prime}+2 \tau^2\right) \right.\right. \\
	&\qquad\left.\left.+
	i \bar{\kappa}_z^\prime \left(k+4 \tau\right)-2 \bar{\kappa}_z^{\prime\prime}+12 |\kappa|^2\bar{\kappa}_z\right) \right) . 
\end{align*}
Here the $A_j$ are $\Tee$ periodic functions of $s$, and $\kappa_z = e^{-i\gamma(s)} \kappa_u$. Note that, if $\gamma$ can be taken as the integral torsion \eqref{eq:gamma}, we have $\kappa_z = \bar{\kappa}_z= \kappa$ the curvature and $\tau$ the torsion of $r_0$.

To degree four, the Hamiltonian \eqref{eq:Hamiltonian} is given explicitly by
\begin{align*}
		H_0 &= -\bar{A}_2 B_0^{-1} Z^2 - A_2 B_0^{-1} \bZ^2 + \left(\tfrac14 k-\tfrac12 \tau \right) Z\bZ ,\\ 
		H_1 &= H_{3,0} Z^3 + \bar{H}_{3,0} \bZ^3 + Z\bZ (H_{2,1} Z + \bar{H}_{2,1} \bZ) ,\\
		H_2 &= H_{4,0} Z^4 + \bar{H}_{4,0} \bZ^4 + Z\bZ (H_{3,1} Z^2 + \bar{H}_{3,1} \bZ^2) + H_{2,2} Z^2 \bZ^2 ,
\end{align*}
where
\begin{align*}
	H_{3,0} &= -B_0^{-3/2}(\bar{A}_3 - \tfrac13 \bar{A}_2\kappa_z ), \\ 
	H_{2,1} &= \tfrac{1}{96} B_0^{-3/2}\left( 56 \bar{\kappa}_z \bar{A}_2 + 5 i \kappa_z B_0^{\prime} - B_0\kappa_z(11 k + 2\tau) + 2 i B_0 \kappa_z^\prime \right) ,\\
	%%%%%%%%%%%%%%%%%%%%%%%%%%%
	H_{4,0} &= -B_0^{-2}\left( \tfrac{1}{24}\kappa_z^2\bar{A}_2 - \tfrac12 \kappa_z \bar{A}_3 + \bar{A}_4 - \tfrac12 B_0^{-1} \bar{A}_2\left( \bar{A}_2\left( k+2\tau \right)-i \bar{A}_2^\prime \right) \right) ,\\
	%%%%%%%%%%%%%%%%%%%%%%%%%%%%%%
	H_{3,1} &= -\tfrac{1}{384}B_0^{-3}\left(\bar{A}_2 \left(48 B_0^{\prime\prime}+8 B_0 \left(8 \left(k^2+ \tau \left(k+\tau\right)\right)-4 i \tau'(s)+7|\kappa_z|^2\right)\right)\right) ,\\
	&\qquad 	+\tfrac{1}{12} i B_0^{-2} \bar{A}_2^\prime \left(k+2 \tau\right) + \tfrac34 B_0^{-2} \bar{\kappa}_z \bar{A}_3+\tfrac{1}{24} B_0^{-2}\bar{A}_2^{\prime\prime} ,\\
	&\qquad +\tfrac{1}{384}B_0^{-2} \kappa_z \left(B_0 \left(\kappa_z \left(7 	k-6 \tau\right)+6 i \kappa_z^\prime\right)-i \kappa_z B_0^{\prime}\right) , \\
	%%%%%%%%%%%%%%%%%%%%%%%%%%%%%%% 
	H_{2,2} &= \tfrac{1}{64}B_0^{-1}k^3 + \tfrac{1}{64} k B_0^{-2}B_0^{\prime\prime} + \tfrac{1}{384}B_0^{-1}|\kappa_z|^2\left(17k+14\tau\right) +B_0^{-3}|A_2|^2(k+2\tau) ,\\
	&\qquad - \tfrac{13}{48} B_0^{-2} \re\left( \kappa_z^2A_2 \right) +\tfrac{14}{384}B_0^{-1} \im\left( \bar{\kappa}_z\kappa_z^\prime \right) - B_0^{-3}\im\left( \bar{A}_2A_2^\prime \right). 
\end{align*}

\end{appendix}

%%%%%%%%%%
%%%% Acknowledgements
%%%%%%%%%%
\section*{Acknowledgements}
The authors acknowledge support of the Simons Foundation through grant \#601972 
``Hidden Symmetries and Fusion Energy." 
Useful conversations with J. Burby, C. Carley, R. Jorge, M. Landreman, R.S. MacKay, W. Sengupta, and E. Rodriguez are gratefully acknowledged.

%%%%%%%%%%
%%%% Bib
%%%%%%%%%%
\bibliography{NearAxisExpansion}
\bibliographystyle{plain}

\end{document}